\documentclass[pra,aps,superscriptaddress,twocolumn,showpacs,nofootinbib]{revtex4-1}
\usepackage{graphicx,color}
\usepackage{amsmath,amsfonts,enumerate,amsthm,amssymb,bbm,mathtools}

\usepackage{hyperref}
\usepackage{multirow}
\usepackage{bbold}
\usepackage{braket}
\usepackage[caption = false]{subfig}
\usepackage{ulem}
\usepackage{scrextend}
\usepackage{blindtext}

\mathchardef\ordinarycolon\mathcode`\:
\mathcode`\:=\string"8000
\begingroup \catcode`\:=\active
  \gdef:{\mathrel{\mathop\ordinarycolon}}
\endgroup

\hypersetup{
   colorlinks=true,         
   linkcolor=blue,          
   citecolor=blue,          
}

\theoremstyle{plain}
\newtheorem{thm}{Theorem}

\newtheorem{propos}[thm]{Proposition}
\theoremstyle{definition}

\theoremstyle{remark}

\def\<{\langle}

\def\P{ {\cal P} }
\def\M{ {\cal M} }

\def\P{ {\cal P} }
\def\O{ {\cal O} }

\def\T{ {\cal T} }
\def\R{ {\cal R} }

\def\I{ \mathbb{1} }

\def\I{ \mathbbm{1} }

\def\>{\rangle}
\def\<{\langle}
\DeclareMathOperator{\Tr}{Tr}

\newcommand{\be}{\begin{equation}}
\newcommand{\ee}{\end{equation}}

\begin{document}

\title{Quantifying athermality and quantum induced deviations from classical fluctuation relations}

\author{Zo\"e Holmes}
\thanks{The first two authors contributed equally to this work.}
\affiliation{Controlled Quantum Dynamics Theory Group, Imperial College London, Prince Consort Road, London SW7 2BW, United Kingdom.}
\author{Erick Hinds Mingo}
\thanks{The first two authors contributed equally to this work.}
\affiliation{Controlled Quantum Dynamics Theory Group, Imperial College London, Prince Consort Road, London SW7 2BW, United Kingdom.}
\author{Calvin  Y.-R. Chen}
\affiliation{Controlled Quantum Dynamics Theory Group, Imperial College London, Prince Consort Road, London SW7 2BW, United Kingdom.}
\author{Florian Mintert}
\affiliation{Controlled Quantum Dynamics Theory Group, Imperial College London, Prince Consort Road, London SW7 2BW, United Kingdom.}

\begin{abstract}
In recent years a quantum information theoretic framework has emerged for incorporating non-classical phenomena into fluctuation relations. Here we elucidate this framework by exploring deviations from classical fluctuation relations resulting from the athermality of the initial  thermal system and quantum coherence of the system's energy supply. In particular we develop Crooks-like equalities for an oscillator system which is prepared either in photon added or photon subtracted thermal states and derive a Jarzynski-like equality for average work extraction. We use these equalities to discuss the extent to which adding or subtracting a photon increases the informational content of a state thereby amplifying the suppression of free energy increasing process. We go on to derive a Crooks-like equality for an energy supply that is prepared in a pure binomial state, leading to a non-trivial contribution from energy and coherence on the resultant irreversibility. We show how the binomial state equality fits in relation to a previously derived coherent state equality and offers a richer feature-set.

\end{abstract}

\maketitle

\section{Introduction}

Thermodynamics, a theory of macroscopic systems at equilibrium, is vastly successful with a diverse range of applications~\cite{Magnetization1,Superconductivity,Cosmology2,WhatIsLife,BioThermo,Chemistry}. This is perhaps somewhat surprising given the prevalence of non-equilibrium states and processes in nature. Underpinning this success is the second law of thermodynamics, an inequality that holds for all equilibrium and non-equilibrium processes alike~\cite{Callen}. Yet the implication of an irreversible flow in the dynamics belies the `arrow of time', since the underlying laws of motion generally define no preferred temporal order~\cite{Jarzynski2}. A resolution to this seeming discrepancy arose in the form of fluctuation theorems, which derive the irreversibility beginning from time-reversal invariant dynamics~\cite{Crooks,TasakiCrooks,Evans,JarzynskiOriginal,Jarzynski2}.

The challenge of generalising fluctuation relations to quantum systems has attracted significant attention in recent years. 
The simplest approach defines the work done on a closed system as the change in energy found by performing projective measurements on the system at the start and end of the non-equilibrium process~\cite{Tasaki,Kurchan,TasakiCrooks,QFT,fluctreview1,fluctreview2}. Extensions to this simple protocol have focused on formulations in terms of quantum channels~\cite{albash, Manzano, Rastegin}, generalisations to open quantum systems~\cite{openingup1,openingup2}, and alternative definitions for quantum work including those using quasi-probabilities~\cite{weakmeasurements,Allahverdyan}, the consistent histories framework~\cite{consistenthistories} and the quantum jump approach~\cite{chap14, quantumjump1,quantumjump2}.
However, these approaches tend to be limited to varying degrees by the unavoidable impact of measurements on quantum systems. By defining quantum work in terms of a pair of projective measurements or continual weak measurements, the role of coherence is attenuated. 

A new framework for deriving quantum fluctuation relations has recently emerged~\cite{aberg,alvaro,hyukjoon,CoherentCrooksChap} which aims to fully incorporate non-classical thermodynamic effects into fluctuation relations by drawing on insights from the resource theory of quantum thermodynamics~\cite{ThermodynamicTransformationsandworkextraction,2ndlaws,workextractionaberg,completestateinterconversion, constraintsbeyondfreeenergy, catalyticcoherence, catalyticworkextraction}. This framework considers an energy conserving and time reversal invariant interaction between an initially thermal system and a quantum \textit{battery}, that is the energy source which supplies work to, or absorbs work from, the system.
This framework can be taken as the starting point to derive Crooks-like relations for a harmonic oscillator battery prepared in coherent, squeezed and Schr\"odinger cat states~\cite{CoherentFluct}.
These new equalities are used both to discuss coherence induced corrections to the Crooks equality and to propose an experiment to test the framework. Furthermore the fluctuation relations give way to an interpretation involving coherent work states, a generalisation of Newtonian work for fully quantum dynamics.
It was proved that the energetic and coherent properties of the coherent work is totally captured in this fluctuation setting~\cite{erick}.

In this paper, we use this new framework to explore deviations from classical fluctuation relations resulting from athermality of the initial thermal system and quantum coherence of the battery. In particular, we start by exploring the effects of athermality by developing Crooks equalities for a quantum harmonic oscillator system which is prepared in a photon added and photon subtracted thermal state. These states have received interest in quantum optics owing to their non-Gaussian and negative Wigner functions~\cite{subphotonExp,subphotonmeanNo,BarnettAddedSubtracted} along with their producibility in lab settings~\cite{subphotonExp,PhotonAddExp,WorkPhotonSub,DemonPhotonSub}. Furthermore, they have been suggested as useful resources in quantum key distribution~\cite{QKDphotonadded}, metrology~\cite{photonaddsubMetrology} and continuous variable quantum computing~\cite{photonaddsubQC12,photonaddsubQC17} and there is growing interest in their thermodynamic properties~\cite{WorkPhotonSub,DemonPhotonSub}.

We then proceed to investigate the role of coherence by deriving a Crooks equality for a battery prepared in pure binomial states. Binomial states can be viewed as analogues of coherent states for finite dimensional systems rather than infinite dimensional oscillators~\cite{GeneralizedCoherentStates,AtomicCoherentStates}, leading to highly non-classical properties~\cite{BinomialStatesStoler,BinomQuasiProb}. While binomial states are harder to produce in lab settings, there have been proposals~\cite{PreparingSpinCS,BinomialStateProduction}.
The derived equality effectively generalises the coherent state Crooks equality of~\cite{CoherentFluct}, incorporating finite sized effects and leading to the coherent state equality in the appropriate limit. Moreover, binomial states quantify a smooth transition between semi-classical regimes and deep quantum regimes by encapsulating both coherent state and multi-qubit fluctuation relations in a single framework. 



\section{Background}

\subsection{Classical Fluctuation Relations}

A system $S$ is initially in thermal equilibrium with respect to Hamiltonian $H_S^i$ at temperature $T$. It is then driven from equilibrium by a variation of Hamiltonian $H_S^i$ to $H_S^f$, doing work $W$ with probability $\P_F(W)$ in the process. This \textit{forwards} process is compared a \textit{reverse} process in which a system thermalised with respect to $H_S^f$ is pushed out of equilibrium by changing $H_S^f$ to $H_S^i$, doing work $-W$ with probability $\P_R (-W)$. The ratio of these two probabilities is known as the Crooks equality~\cite{Crooks},
\begin{equation}\label{eq:Crooks}
    \frac{\P_F (W)}{\P_R (-W)} = \exp \left(\beta (W - \Delta F) \right) \; ,
\end{equation}
where $\Delta F$ is the equilibrium Helmholtz free energy difference and $\beta$ is the inverse temperature $1/k_B T$. 


The Crooks equality is a generalisation of the second law of thermodynamics. As a corollary to Crooks equality, one can derive the Jarzynski equality~\cite{originaljarzynski}, which reads
\begin{equation}\label{eq:Jarzynski}
    \left\langle \exp \left(-\beta W\right) \right\rangle =  \exp \left(- \beta \Delta F \right)  \; .
\end{equation}
Finally, using Jensen's inequality~\cite{JensenInequality}, one arrives at the second law of thermodynamics in its formulation as a bound for the average extractable work $\< W_{\rm ext} \> \leq - \Delta F$.
The Jarzynski equality has been used to calculate free energy changes for highly complex systems~\cite{ComplexJarzynski} such as unravelling of proteins~\cite{JarzynskiProtein},
and as a theoretical tool to re-derive two of Einstein's key relations for Brownian motion and stimulated emission~\cite{FredGittes}.


\subsection{Fully Quantum Fluctuation Relations}

Our starting point is a global ''fully quantum fluctuation theorem" from \cite{aberg}, a more general relation than that explicated in~\cite{CoherentFluct, CoherentCrooksChap,erick}, which can be used to derive a whole family of quantum fluctuation relations. A defining property of quantum systems is their ability to reside in superpositions of states belonging to different energy eigenspaces, a property often referred to simply as coherence. The quantum framework we present here carefully tracks the changes in these energetic coherences. 

Changing the Hamiltonian of a system typically requires doing work or results in the system performing work and thus every fluctuation relation, at least implicitly, involves an energy source which supplies or absorbs this work. While often not explicitly modelled, the dynamics of the energy supply can contribute non-trivially to the evolution of the driven system. Thus, to enable a more careful analysis of the energy and coherence changes of the system, we consider an \textit{inclusive}\footnote{This is in contrast to the \textit{exclusionary} picture of the original Crooks and Jarzynski equalities.} approach~\cite{deffjarzinclusive, aberg, alvaro, erick, hyukjoon,CoherentFluct} which introduces a \textit{battery} and assumes the system (S) and battery (B) evolve together under a time independent Hamiltonian $H_{SB}$.

To realise an effective change in system Hamiltonian from $H_S^i$ to $H_S^f$ with a time independent Hamiltonian, we assume a Hamiltonian of the form
\begin{equation}
    H_{SB} =\I_S \otimes H_B + H_S^i\otimes \Pi_B^i +  H_S^f\otimes \Pi_B^f 
\end{equation}
where $H_B$ is the battery Hamiltonian and $\Pi_B^i$ and $\Pi_B^f$ are projectors onto two orthogonal subspaces, $R_i$ and $R_f$, of the battery's Hilbert space. We assume the battery is initialised in a state in subspace $R_i$ only and evolves under a unitary $U$ to a final state in subspace $R_f$ only, such that the system Hamiltonian is effectively time dependent, evolving from $H_S^i$ to $H_S^f$. 

To ensure that the energy supplied to the system is provided by the battery we 
require the dynamics to be energy conserving such that $[U,H_{SB}]=0$.
We further assume that $U$ and $H_{SB}$ are time-reversal invariant with $U = \mathcal{T}(U) $ and $ H_{SB} = \T (H_{SB})$. The time-reversal~\cite{timereversalop, CrooksTimeRev} operation $\T$ is defined as the transpose operation in the energy eigenbasis of the system and battery.

The most general process that can be described by a fluctuation relation within the inclusive framework involves preparing the system and battery in an initial state $\rho$, evolving it under the propagator $U$ and then performing a measurement on the system and battery which can be represented by the measurement operator $X$. The outcome of this measurement is quantified by 
\begin{equation}\label{eq:Qquantity}
    \mathcal{Q}(X|\rho) := \Tr \left[X U \rho U^\dagger  \right] \ ,
\end{equation}
which can capture a number of different physical properties. 
For example, if the measurement operator $X$ is chosen to be an observable then $\mathcal{Q}(X|\rho)$ is the expectation value of the evolved state $U \rho U^\dagger$. Whereas, if the measurement operator is chosen to be some state $\rho'$, corresponding to the binary POVM measurement $\{ \rho', \I - \rho' \}$, then $Q(\rho'| \rho)$ captures a transition probability between the state $\rho$ and $\rho'$ under the evolution $U$. 

\begin{figure}
  \centering
{\includegraphics[width=0.92\linewidth]{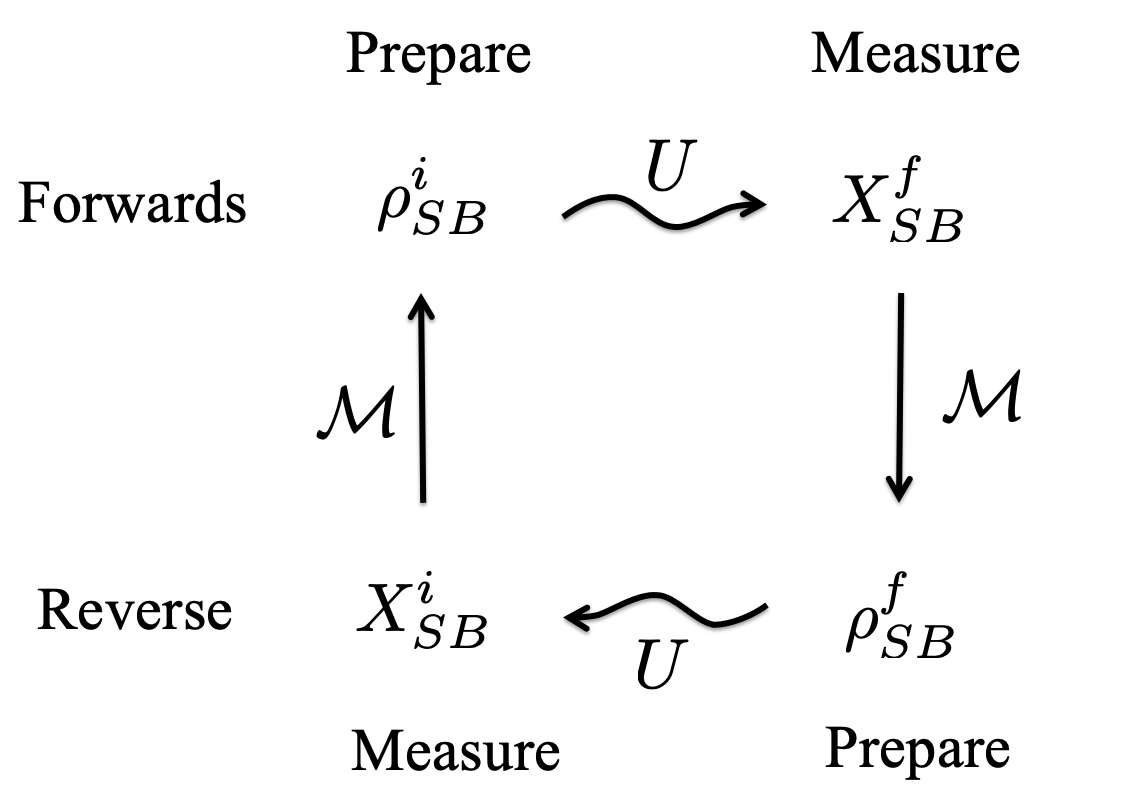}}
\caption{ \textbf{Relation between prepared states and measurements.} In the forwards (reverse) process, the state $\rho_{SB}^i = \rho_S^i \otimes \rho_B^i$ $\left(\rho_{SB}^f = \rho_S^f \otimes \rho_B^f\right)$ is prepared, it evolves under $U$ as indicated by the wiggly arrow, and then the measurement $X_{SB}^f = X_S^f \otimes X_B^f$ $\left(X_{SB}^i = X_S^i \otimes X_B^i\right)$ is performed. As indicated by the solid lines, the measurements $X_{SB}^i$ and $X_{SB}^f$ are related to the states $\rho_{SB}^i$ and $\rho_{SB}^f$ respectively by the mapping $\M$, defined in Eq.~\eqref{eq:Mapping}.}
\label{fig:Mapping}
\end{figure}

The global fluctuation relation relates $\mathcal{Q}(X_{SB}^f|\rho_{SB}^i)$ of a forwards process to $\mathcal{Q}(X_{SB}^i|\rho_{SB}^f)$ of a reverse process. For our purposes we will assume that the system and battery are initially uncorrelated in both the forwards and reverse processes, i.e. 
\begin{equation}\label{eq:UncorStates}
    \rho_{SB}^i = \rho_{S}^i \otimes \rho_B^i \ \ \ \text{and} \ \ \  \rho_{SB}^f = \rho_{S}^f \otimes \rho_B^f
\end{equation}
and suppose that independent measurements are made on the system and battery such that the measurement operator can be written in a separable form, i.e.
\begin{equation}\label{eq:UncorMeasure}
    X_{SB}^i = X_{S}^i \otimes X_B^i \ \ \ \text{and} \ \ \  X_{SB}^f = X_{S}^f \otimes X_B^f \ . 
\end{equation}
The global fluctuation relation holds for measurement operators and states related by the mapping $\M$ defined as
\begin{align}\label{eq:Mapping}
    &\rho_S^k = \M( X_S^k) \propto \T \left( \exp\left(-\frac{\beta H_S^k}{2}\right) X_S^k \exp\left(-\frac{\beta H_S^k}{2}\right) \right) \\
    &\rho_B^k = \M( X_B^k) \propto \T \left( \exp\left(-\frac{\beta H_B}{2}\right) X_B^k \exp\left(-\frac{\beta H_B}{2}\right) \right) 
\end{align}
for $k = i, f$. This mapping arises naturally when one relates a forward and a reverse quantum process in the inclusive framework. When a measurement operator is a projection onto an energy eigenstate then the state related by the mapping, Eq.~\eqref{eq:Mapping} is an energy eigenstate. Conversely, when no measurement is performed, i.e. $X = \I$, the corresponding state is a thermal state. However, in general the mapping is non trivial and essential to capture the influence of quantum coherence and athermality.

For the uncorrelated initial states and measurement operators related by the mapping $\M$ the global fluctuation relation~\cite{aberg, CoherentFluct, erick}
can be written as 
\begin{equation}\label{eq:GlobalCrooks}
    \frac{\mathcal{Q}(X_{SB}^f|\rho_{SB}^i)}{\mathcal{Q}(X_{SB}^i|\rho_{SB}^f)}
    = \exp\left(\beta (\Delta \tilde{W} -\Delta \tilde{F})\right) \, ,
\end{equation}
in terms of the quantum generalisation
\begin{align}\label{eq:GenFreeEnergy}
    \Delta \tilde{F} := \tilde{E}(\beta,H_S^f, X^f_S) - \tilde{E}(\beta,H_S^i, X^i_S)  
\end{align}
of the change in free energy,
as well as a quantum generalisation of the work
\begin{align}\label{eq:GenWork}
    \Delta \tilde{W} := \tilde{E}(\beta,H_B, X^i_B) - \tilde{E}(\beta,H_B, X^f_B)  \; .
\end{align}
supplied by the battery.
The function
\begin{equation}
    \tilde{E}(\beta,H, X) := -\frac{1}{\beta} \, \ln \left( \Tr \left[\exp\left(-\beta H \right) X \right]\right) \; 
\end{equation}
is an \textit{effective potential} that specifies the relevant energy value within the fluctuation theorem context. When the measurement operator is equal to the identity operation the effective potential, $\tilde{E}(\beta, H, \I) $, is equal to the free energy with respect to Hamiltonian $H$ and thus $\Delta \tilde{F}$ reduces to the usual Helmholtz free energy. Conversely, for a projector onto an energy eigenstate the effective potential, $\tilde{E}(\beta,H,|E_k\>\<E_k|)$, is the corresponding energy $ E_k$ from which we regain the classical work term using a two point projective measurement scheme. More generally, when restricting to projective measurement operators, the function $\beta \tilde{E}(\beta,H,|\psi\>\<\psi|) $ is a cumulant generating function in the parameter $\beta$ that captures the statistical properties of measurements of $H$ on $|\psi\>$~\cite{erick}.

We regain the Crooks equality from this global fluctuation relation for a thermal system and a battery with a well defined energy. Specifically, in the forwards process the system is prepared in a thermal state
\begin{equation}
    \gamma_S^{i}  \propto \exp\left(-\beta H_S^i\right) \, 
\end{equation}
and we consider the probability to observe the battery to have energy $E_f$ having prepared it with energy $E_i$, that is transition probabilities of the form
\begin{equation}\label{eq:TransProb}
    \begin{aligned}
    \P(E_f | \gamma_S^i, E_i) &:= \mathcal{Q}\left(\I_S \otimes |E_f\>\<E_f|\, \bigg{|} \, \gamma_S^i \otimes
    |E_i\>\<E_i|\right) \ . 
    \end{aligned}
\end{equation}
In this classical limit, the global fluctuation relation reduces to 
\begin{equation}\label{eq:ClassLimitGFR}
    \frac{\P(E_f | \gamma_S^i, E_i)}{\P(E_i | \gamma_S^f, E_f)}
    = \exp\left(\beta(W -\Delta F)\right) \,
\end{equation}
where  $W := E_i - E_f$ is the negative change in energy of the battery and thus, due to global energy conservation, equivalent to the work done on the system. If we additionally assume that the dynamics of the system and battery do not depend on the initial energy of the battery, then using this \textit{energy translation invariance} assumption which we explicitly define in Section~\ref{sec:Jarz}, one is able to regain all classical and semi-classical fluctuation results~\cite{aberg}. The global fluctuation relation is thus a genuine quantum generalisation of these relations and inherits their utility.




\medskip

In this manuscript we use the global fluctuation relation, Eq.~\eqref{eq:GlobalCrooks}, to quantify deviations from the classical Crooks relation resulting from athermality of the initial thermal system and quantum coherence of the battery. Specifically, to probe the impact of preparing the system in imperfectly thermal states, we derive in Section~\ref{sec:PhotonAddSub} a Crooks-like relation for a system that is prepared in a photon added or a photon subtracted thermal state. In section~\ref{sec:GenCoherentStates}, we investigate the deviations generated by coherence in the battery by deriving a Crooks equality for binomial  states of the battery.

\section{Results}

\subsection{Photon added and subtracted thermal states}\label{sec:PhotonAddSub}

Photon added and subtracted states are non-equilibrium states generated from a thermal state by, as the name suggests, either the addition or the subtraction of a single photon. Considering a single quantised field mode with creation and annihilation operators $a^\dagger$ and $a$ and Hamiltonian $H$, the photon added thermal state can be written as 
\begin{equation}\label{eq:add}
    \gamma_{H}^+ \propto a^\dagger \exp\left(- \beta H\right) a \; 
\end{equation}
and the photon subtracted thermal state as
\begin{equation}\label{eq:sub}
    \gamma_{H}^- \propto a \exp\left(- \beta H \right) a^\dagger \; .
\end{equation}
The states $\gamma_{H}^+$ and $\gamma_{H}^-$ are diagonal in the energy eigenbasis and therefore are classical in the sense that they are devoid of coherence. Nonetheless, they are non-Gaussian and have negative Wigner functions~\cite{photonaddedNonClas97,photonaddedNonClas06,photonaddedNonClas07,photonsubNonClas09, photonsubNonClas17}, traits which are considered non-classical in the context of quantum optics. 

Moreover, the addition or subtraction of a photon from a thermal state has a rather surprising impact on the number of photons in the state:
In particular, adding a photon to a thermal state of light, which contains on average $\bar{n}$ photons, increases the expected number of photons in the state to $2 \bar{n} + 1$~\cite{subphotonExp,subphotonmeanNo,BarnettAddedSubtracted}.
Similarly, subtracting a photon from a thermal state doubles the expected number of photons to $2 \bar{n}$. Thus, counter-intuitively, adding or subtracting a \textit{single} photon to a thermal state \textit{substantially increases} the expected number of photons in the state.

In line with standard nomenclature we will refer to \textit{photon} added and subtracted thermal states throughout this paper; however, the modes in Eq.~\eqref{eq:sub} and Eq.~\eqref{eq:add} could naturally refer to any \textit{boson}. Experimental techniques for generating photon added~\cite{PhotonAddExp} and subtracted~\cite{subphotonExp} thermal states are well established and methods are currently being developed for the preparation of \textit{phonon} added states~\cite{PhononAddedCoherent}.

To illustrate the deviations from classical thermodynamics induced by the addition (subtraction) of a single photon we derive a Crooks-like relation characterised by replacing the initially thermal system of the standard setting quantified by the Crooks equality, with a system in a photon added (subtracted) thermal state. That is, for the photon added (+) and photon subtracted (-) equalities we suppose that the system is prepared in the states
\begin{align}
     &\rho_S^i = \gamma_{i}^\pm  \ \ \ \text{and} \ \ \ \rho_S^f = \gamma_{f}^\pm
\end{align}
at the start of the forwards and reverse processes respectively, where to simplify notation we have introduced the shorthand $\gamma_{k}^\pm \equiv \gamma_{H_S^k}^\pm$. 

In analogy to the classical Crooks relation, we quantify the work supplied to the system when the photon added (subtracted) thermal system is driven by a change in Hamiltonian. For concreteness we assume here that the system is a quantum harmonic oscillator with initial and final Hamiltonians given by 
\begin{equation}
    H_S^k := \hbar  \omega_k \left( a_k^\dagger a_k + \frac{1}{2}\right)  \; ,
\end{equation}
for $k =i$ and $k= f$, such that the system is driven by a change in its frequency from $\omega_i$ to $\omega_f$. 
As energy is globally conserved, the work supplied to the system is given by the change in energy of the battery and therefore the probability distribution for the work done on the system can be quantified by transition probabilities between energy eigenstates of the battery. Specifically, in the forward process we consider the probability to observe the battery to have energy $E_f$ having prepared it with energy $E_i$ and vice versa in the reverse. We do not need to make any specific assumptions on the battery Hamiltonian $H_B$ to quantify such eigenstate transition probabilities and therefore $H_B$ may be chosen freely.

In contrast to the usual Crooks relation, the photon added (subtracted) Crooks relations depends on the average number of photons in the photonic system after the driving process. This arises from the mapping $\M$ between the measurement operators and the initial states 
following Eq.~\eqref{eq:Mapping}.
As shown explicitly in Appendix~\ref{ap:DerivationPhotonAddSub}, on inverting Eq.~\eqref{eq:Mapping} we find that for the photon added equality the measurement operators $X_S^i$ and $X_S^f$ 
are given by
\begin{equation}
     X_S^k = a_k^\dagger a_k := N_k \ \ \ \text{for} \ \ \ k=i, \, f \, ;
\end{equation}
and for the subtracted equality
they are given by
\begin{equation}
     X_S^k = a_k a_k^\dagger = N_k + 1 \ \ \ \text{for} \ \ \ k=i, \, f \; .
\end{equation}
That is, in both cases, they are given in terms of the number operator $N_k$ only.

Given this form for the measurement operators, it follows that the photon added and subtracted Crooks relations quantify the expected number of photons in the system at the end of the driving process as well as the change in energy of the system.
For example, for the forwards process of the photon added Crooks equality $\mathcal{Q}$, as defined in Eq.~\eqref{eq:Qquantity}, is equal to
\begin{equation}\label{eq:FactoredQ}
\begin{aligned}
    \mathcal{Q}\left(N \otimes |E_f\>\<E_f| \, \bigg{|} \, \gamma_i^+ \otimes |E_i\> \<E_i| \right) \\ = n(E_f | \gamma_i^+, E_i) \, \P(E_f | \gamma_i^+, E_i) \ ,
\end{aligned}
\end{equation}
where $\P$ is the transition probability of the battery from energy $E_i$ to $E_f$  conditional on preparing the system in a photon added thermal state, as defined in Eq.~\eqref{eq:TransProb}, and $ n(E_f | \gamma_i^+, E_i)$ is the average number of photons in the system at the end of this driving process. Similar expressions to Eq.~\eqref{eq:FactoredQ} are obtained for the reverse process of the photon added equality and both the forwards and reverse processes of the photon subtracted equality. 

As we are considering transition probabilities between energy eigenstates of the battery, the generalised energy flow term $\Delta \tilde{W}$ reduces to the work done on the system as in Eq.~\eqref{eq:ClassLimitGFR}. However, as derived explicitly in Appendix~\ref{ap:DerivationPhotonAddSub}, the generalisations of the free energy term, Eq.~\eqref{eq:GenFreeEnergy}, $\Delta \tilde{F}^+$ and $\Delta \tilde{F}^-$ for the photon added and subtracted equalities respectively, evaluate to
\begin{equation}
    \Delta \tilde{F}^\pm =  2 \Delta F \pm \Delta E_{\mbox{\footnotesize vac}} \, .
\end{equation}
In the above $\Delta F$ is the change in free energy associated with the change in Hamiltonian from $H_S^i$ to $H_S^f$ and we have introduced $\Delta E_{\mbox{\footnotesize vac}}$,
\begin{equation}
    \Delta E_{\mbox{\footnotesize vac}} := \frac{1}{2}\hbar \omega_f - \frac{1}{2} \hbar \omega_i \; ,
\end{equation}
as the difference between the initial and final vacuum energies of photonic system. 

\begin{figure}
  \centering
{\includegraphics[width=\linewidth]{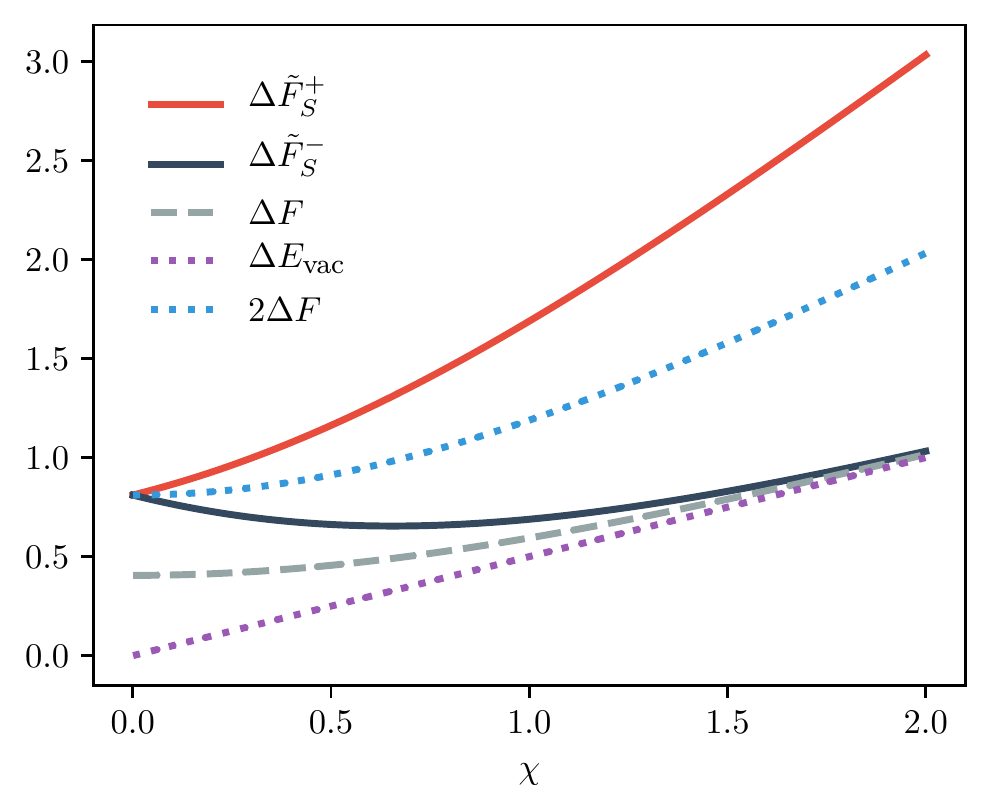}}
\caption{\textbf{Generalised Free Energies.} The solid red and dark blue lines show the generalised free energy, $\Delta F^+$ and $\Delta F^-$, of the oscillator system for the photon added and photon subtracted equalities respectively. These are plotted as a function of $\chi =  \beta \hbar \omega_i/2$, the ratio between the initial vacuum fluctuations, $\hbar \omega_i/2$, and the thermal fluctuations, $ k_B T$, a measure which quantifies the temperature and thus effectively delineates the classical and quantum regimes. The grey dashed line is the usual change in energy $\Delta F$. The dotted lines indicate the contribution of $\Delta E_{\mbox{\tiny vac}}$ (purple) and $2 \Delta F$ (light blue) to $\Delta F^+$ and $\Delta F^+$. In this plot we suppose $\hbar \omega_f = 1.5 \hbar \omega_i$ and energies are given in units of $k_B T$.}
\label{fig:PhotonAddSubPlot}
\end{figure}

In the classical limit where $\hbar$ tends to zero the contribution from the energy of the vacuum state, $\Delta E_{\mbox{\footnotesize vac}}$, vanishes and $\Delta F^+$ and $\Delta F^-$ both tend to $2 \Delta F$. This behaviour can be explained by the observation in~\cite{BarnettAddedSubtracted} that the photon probability distributions for photon added and subtracted states have the same functional form but while the photon subtracted distribution starts at $n = 0$, that is in the vacuum state, the photon added distribution starts at $n=1$, and therefore has no vacuum contribution, a shift which becomes increasingly insignificant for higher temperatures. Conversely, as shown in Fig.~\ref{fig:PhotonAddSubPlot}, in the low temperature quantum limit the contribution of the energy of the vacuum state generates sizeable deviations between the generalised free energy terms for the photon added and subtracted cases. Specifically, while $\Delta F_S^-$ tends to $\Delta F$ in agreement with the standard classical Crooks relation, we find that $\Delta F^+$ is substantially larger than $2 \Delta F$. This is due to the fact that in the low temperature limit the photon subtracted thermal state and normal thermal state both tend to the vacuum state, whereas the photon added thermal state tends to a single photon Fock state. In all limits $\Delta F^+$ and $\Delta F^-$ are larger than $\Delta F$ indicating that the addition and subtraction of a photon increases the energy and information content of a thermal state thereby increasing the extractable work from the state. Similar phenomena have been observed elsewhere in the context of work extraction protocols~\cite{WorkPhotonSub} and Maxwell demons~\cite{DemonPhotonSub}.

The final photon added (+) and photon subtracted (-) Crooks equality can be written as
\begin{equation}\label{eq:addsubCrooks}
    \frac{\P(E_f | \gamma_i^\pm, E_i)}{\P(E_i | \gamma_f^\pm, E_f)} = \R_\pm(W) \exp\left(\beta\left(W -  2\Delta F \mp \Delta E_{\mbox{\footnotesize vac}}\right)\right) \; .
\end{equation}
The prefactor $\R_\pm(W)$ quantifies the ratio of the number of photons measured in the system at the end of the reverse process over the number of photons measured at the end of the forwards process\footnote{As a result, the prefactor is only defined when both the numerator and denominator of Eq.~\eqref{eq:Prefactors} are both positive quantities.}. As shown in Appendix~\ref{ap:DerivationPhotonAddSub}, the prefactors $\R_+(W)$ and $\R_-(W)$ can be written as 
\begin{equation}
    \begin{aligned}\label{eq:Prefactors}
        &\R_\pm(W) = \frac{\omega_f}{\omega_i} \frac{\hbar \omega_f(2 \bar{n}_f + k_\pm) +  W + \Delta E_{\mbox{\tiny vac}}}{\hbar \omega_i \left(2 \bar{n}_i + k_\pm^{-1}\right) - W - \Delta E_{\mbox{\tiny vac}}} 
\end{aligned}
\end{equation}
with $k_+ = 1$ and $k_- = \frac{\omega_i}{\omega_f}$ and where $\bar{n}_k$ is the average number of photons in a thermal state with frequency $\omega_k$. It is worth noting that $\R_\pm(W)$ implicitly depends on the free energy of the initial and final Hamiltonians because $\hbar \omega_k (\bar{n}_k + \frac{1}{2})$ is the average energy of a thermal photonic state with frequency $\omega_k$, which by definition, is equal to the sum of free energy and entropy of the state. 

\begin{figure}
\centering
{\includegraphics[width=\linewidth]{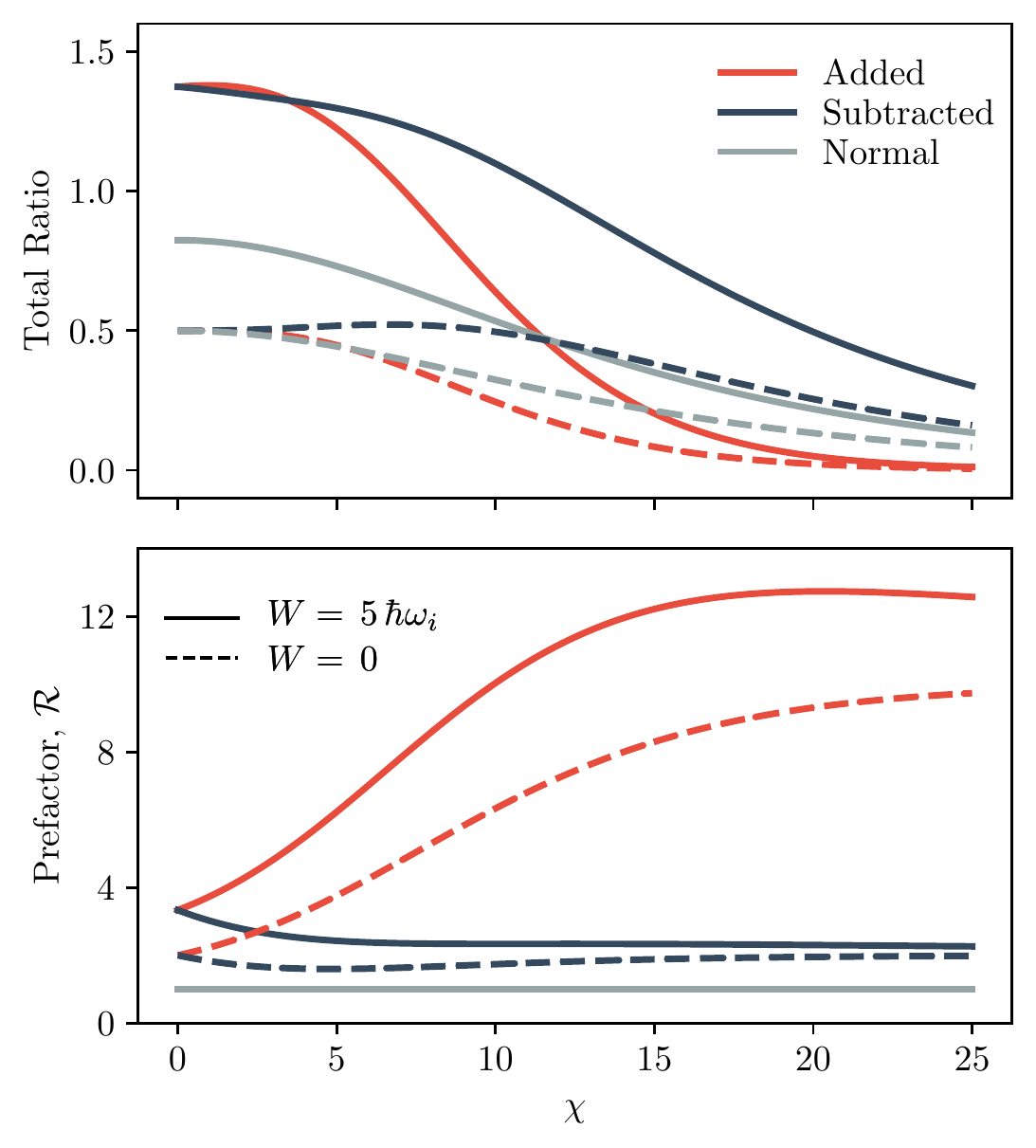}}
\caption{\textbf{Predicted ratio and $\R$ prefactor.} The upper figure plots the predicted ratio of the forwards and reverse transition probabilities, i.e. the right hand side of Eq.~\eqref{eq:addsubCrooks}, for the photon added (subtracted) Crooks equality as a function of $\chi =  \beta \hbar \omega_i/2$. The lower figure plots $\R$ as a function of $\chi$. The red (blue) lines indicates the photon added (subtracted) case and the grey lines indicate the equivalent classical limit. That is, in the upper plot the grey line is the right hand side of the classical Crooks equality, Eq.~\eqref{eq:Crooks}, and in the lower plot the grey line is $\R = 1$. The solid lines plot the case $W =2 \hbar \omega_i$ and the dashed lines, $W = 0$. Here we suppose $\hbar \omega_f = 5 \hbar \omega_i$.}
\label{fig:PhotonAddSubTotalRatio}
\end{figure}

The classical Crooks equality implies that driving processes which require work and decrease free energy are exponentially more likely than processes which produce work and increase free energy, thus quantifying the irreversibility of non-equilibrium driving processes. Given that the generalised free energy terms $\Delta \tilde{F}^+$ and $\Delta \tilde{F}^-$ are greater than the usual change in free energy $\Delta F$, it is tempting to conclude that athermality of the initial system can strengthen irreversibility by amplifying the suppression factor of free energy increasing processes. However, the presence of the prefactor $\R$ in Eq.~\eqref{eq:addsubCrooks}, which depends on both the work done during the driving process and implicitly the initial and final free energies of the system, makes it harder to draw clear cut conclusions. 

To aid comparison between the athermal and thermal cases, in Fig.~\ref{fig:PhotonAddSubTotalRatio} we plot the total predicted ratio of the forwards and reverse processes for the photon added and subtracted Crooks relations, that is the right hand side of Eq.~\eqref{eq:addsubCrooks}, and compare them to the equivalent prediction of the classical relation, Eq.~\eqref{eq:Crooks}. We similarly plot the prefactors $\R_+$ and $\R_-$. As the prefactor $\R$ does not appear in the classical Crooks relation, Eq.~\eqref{eq:Crooks}, we can say that $\R$ is effectively equal to 1 in the limit of a perfectly thermal system. For concreteness, we here consider a forwards process where the oscillator frequency is doubled, increasing the system's free energy. We plot the ratio and $\R$ as a function of $\chi := \frac{\beta \hbar \omega}{2 }$, the ratio of vacuum energy to thermal energy, a measure which delineates between quantum and thermodynamic regimes. 


As shown in Fig.~\ref{fig:PhotonAddSubTotalRatio}, the interplay between the prefactors $\R_\pm$, which are greater than the classical limit of 1, and the terms $\exp(-\beta \Delta \tilde{F}_\pm)$, which are smaller than $\exp(- \beta \Delta F)$, leads to a rich spectrum of deviations from the classical Crooks relation. For example, while the prefactor $\R_+$ for the photon added case is substantially greater than 1 in the low temperature limit, the total predicted ratio is smaller than for the photon subtracted case. This is because the large value of $R_+$ is exponentially suppressed by $\Delta \tilde{F}_+$ which is substantially larger than $\Delta \tilde{F}_-$ and $\Delta F$, as shown in Fig.~\ref{fig:PhotonAddSubPlot}, due to the contribution of the change in vacuum energy. Thus we conclude that for the photon added relation, irreversibility is milder in the quantum limit due to the contribution of the energy of the vacuum state, a phenomena which was also observed in \cite{CoherentFluct}.


In the high temperature classical limit one might expect adding or subtracting a single photon to a thermal state containing on average a large number of photons would have a negligible effect.
Indeed this is what we see for processes in which no work is performed on the system since in the high temperature limit the prefactor $\R_\pm(0)$ reduces to $\exp(\beta \Delta F)$. However, interestingly for work requiring processes we do see large deviations from the usual classical Crooks relation in the classical limit. We attribute this to the fact that adding or subtracting a photon from thermal light effectively doubles the mean photon number the state, and therefore the net effect can be substantial even for high temperature states as they contain larger numbers of photons. 

More generally, for all temperatures and for both the photon added and subtracted relations, we find that the larger the work done on the system, the larger the predicted ratio. This confirms that even when the initial states are photon added or subtracted thermal states, processes which require work are exponentially more probable than processes that generate work. 

\subsection{Binomial states}\label{sec:GenCoherentStates}

In the previous section we showed how the \textit{athermality} of the initial \textit{system}, due to the addition or subtraction of a single photon, induces rich  deviations from the classical Crooks relation. Here we complement this analysis by exploring how quantum features can be introduced through the \textit{coherence} of the \textit{battery}.
The quantum fluctuation relations are well characterized for coherent states of the battery~\cite{CoherentFluct} which have close-to-classical properties.
In the following we will consider binomial states, which provide a well-defined transition between coherent states of a quantum harmonic oscillator, and highly quantum mechanical states such as a state of an individual qubit.

Binomial states are pure states of the form 
\begin{equation}\label{eq: binomial state}
    |n,p\> = \sum_{k=0}^n  \sqrt{{n \choose k}p^k(1-p)^{n-k}} \, e^{i \phi_k} |k\>,
\end{equation}
whose properties have been extensively studied in the field of quantum optics~\cite{BinomialStatesStoler,AtomicCoherentStates,BarnettAddedSubtracted,MultinomialStates}. Binomial states are non-classical states with finite support and exhibit sub-Poissonian statistics~\cite{BinomialStatesStoler,BarnettAddedSubtracted}, squeezing of quadratures~\cite{BinomialStatesStoler} and are highly non-classical both in terms of their coherent properties and the negativity of their Wigner function~\cite{BinomQuasiProb}. They can be thought of as an $n$-qubit tensor product  $|p\>^{\otimes n}$ of the states $|p\> = \sqrt{1-p}|0\> + \sqrt{p}|1\>$. The states $|n,p\>$ and $|p\>^{\otimes n}$ are related by an energy-preserving unitary rotation. This is important as the effective potential $\tilde{E}$ is invariant under energy conserving unitaries, implying that as far as the fluctuation theorem is concerned, they are interchangeable.
In the limit that $n$ tends to infinity they approach the regular coherent states and the opposite limiting case $n=1$ corresponds to the deep quantum regime.

Binomial states find use owing to their nice analytical properties. For instance, the commonly encountered spin-coherent states are particular examples of binomial states~\cite{GeneralizedCoherentStates,MultinomialStates,VisualizingSpin,SpinCoherent}. Spin-coherent states belong to a class of generalised coherent states that allow for different displacement operators, in this case of the form $D(\alpha) = \exp(\alpha S_+ + \alpha^* S_-)$  where $S_\pm$ are the spin-raising and lowering operators~\cite{AtomicCoherentStates,VisualizingSpin,GeneralizedCoherentStates}. Proposals for the generation of binomial states have been developed in atomic systems~\cite{PreparingSpinCS,BinomialStateProduction} and they have been suggested as analogues to coherent states for rotational systems~\cite{RotationSpinState,BinPotential}. These examples indicate that binomial states are of natural physical interest.

\medskip

In what follows, we assume the battery is a harmonic oscillator\footnote{One could also consider a finite Hamiltonian, however for complete generality, decoupling the dimension of the Hamiltonian and the support of the state proves useful.}, $H_B = \hbar \omega (a^\dagger a + \frac{1}{2})$, but do not make any specific assumptions on the initial and final system Hamiltonians. We assume the system is prepared in a standard thermal state and consider transitions between two binomial states of the battery. More specifically, here the battery measurement operators are chosen as the projectors
\begin{align}\label{eq:bin operators a}
X_B^k  = |n_k,p_k \>\< n_k , p_k |  \ \ \ \text{for} \ \ \ k=i, f \; .
\end{align}
which, given the mapping $\M$ in Eq. (\ref{eq:Mapping}), fixes the preparation states. As shown in Appendix \ref{ap:Binomial}, we find that the prepared states are the binomial states,
\begin{align}\label{eq:bin operators}
&\rho_B^k  = |n_k,\tilde{p}_k \>\< n_k , \tilde{p}_k |    \ \ \ \text{with,} \\ \ \
&\tilde{p} = \frac{p e^{-\beta \hbar \omega}}{p e^{-\beta \hbar \omega} + q} \ \ \ \text{and} \ \ \ \tilde{q} = \frac{q}{p e^{-\beta \hbar \omega} + q} \; ,
\end{align}
with $q = 1-p$ and for $k=i, f$. Thus we see that the mapping $\M$ preserves binomial statistics but leads to a distortion factor due to the presence of coherence. Since $\tilde{p}$ is always less than $p$, this distortion from $\M$ lowers the energy of the prepared state as compared to the equivalent measured state, with its energy vanishing in low temperature limit.

There exist two clear distinct physical regimes corresponding to different battery preparation and measurement protocols. In the \textit{realignment} regime, we fix the system size $n$ and consider transition probabilities between rotated states. Conversely, the \textit{resizing} regime quantifies transition probabilities between states of different `sizes', that is states with different supports but fixed alignment in the Bloch sphere. 
For the realignment regime, the prepare and measure protocols are as follows.
\begin{center}
    Forwards: The battery $B$ is prepared in the state $|n,\tilde{p}_i \>$ and measured in $|n,p_f\>$ \\
    Reverse: The battery $B$ is prepared in the state $|n,\tilde{p}_f\>$ and measured in $|n,p_i\>$.
\end{center}
 While for the resizing regime, where we fix $p$ and vary $n$, we have the prepare and measure protocol
\begin{center}
    Forwards: The battery $B$ is prepared in the state $|n_i,\tilde{p} \>$ and measured in $|n_f,p\>$. \\
    Reverse: The battery $B$ is prepared in the state $|n_f,\tilde{p}\>$ and measured in $|n_i,p\>$.
\end{center}

In the qubit picture, for a system of $N$ qubits the realignment regime amounts to fixing the number of battery qubits with coherence to precisely $n$ while changing the polarisation $p_k$ of each of these $n$ qubits concurrently. Similarly, the resizing regime corresponds to fixing the polarisation and changing the number of non classical qubits. More precisely, we can write
  \begin{equation}
    \ket{n_k, p_k} \equiv |p_k\>^{\otimes n_k} \otimes |0\>^{\otimes N-n_k} \ \ \ \text{for} \ k=i,f \; 
\end{equation}
where in the first regime $n_k$ is kept fixed while $p_k$ is varied and vice versa for the second. In the context of spin-coherent states, the first regime corresponds to a battery that remains a spin-$\frac{n}{2}$ system but whose orientation varies, while the second amounts to changing the magnitude of the spin while fixing the orientation.

\medskip

The key quantity in the fluctuation relation is the generalised work flow, the derivations of which can be found in Appendix \ref{ap:Binomial}. In these processes, the generalised work flow in the realignment regime and resizing regimes, $\Delta \tilde{W}_{\mbox{\tiny align}}$ and $\Delta \tilde{W}_{\mbox{\tiny size}}$ respectively, take the form
\begin{align}\label{eq:flow regime a}
    \beta \Delta \tilde{W}_{\mbox{\tiny align}}&=  n  \left(\ln \frac{p_f}{\tilde{p}_f} - \ln \frac{p_i}{\tilde{p}_i} \right) \\
     \beta \Delta \tilde{W}_{\mbox{\tiny size}} &= (n_f -n_i)  \left(\ln \frac{p}{\tilde{p}} + \beta \hbar \omega \right), \label{eq:flow regime s}
\end{align}
provided both $p_i$ and $p_f$ are non-zero.
These capture the temperature-dependent distortion of the binomial states due to $\M$. While the generalised work flow in the realignment regime smoothly varies with its free parameters, in the resizing regime the energy flow is discretised. The binomial state Crooks relations corresponding to the realignment and resizing regimes follow upon insertion of the generalised work flow terms, Eq.~\eqref{eq:flow regime a} and Eq.~\eqref{eq:flow regime s}, into the global fluctuation relation, Eq.~\eqref{eq:GlobalCrooks}, when restricted to binomial state preparations specified in Eq.~\eqref{eq:bin operators}.

In the high temperature limit,  $\beta \hbar \omega \ll 1$, we can truncate the power series of $\Delta \tilde{W}$ to second order for sufficient accuracy, which gives
\begin{align}\label{eq:flow regime 1}
\beta \Delta \tilde{W}_{\mbox{\tiny align}} &\approx  \beta\hbar \omega n \left( p_i - p_f\right) -  \frac{(\beta \hbar \omega)^2}{2} (\sigma_i^2 - \sigma_f^2) \\
\beta \Delta \tilde{W}_{\mbox{\tiny size}} &\approx \beta \hbar \omega (n_i - n_f) p - \frac{(\beta \hbar \omega)^2}{2} (n_i - n_f) \sigma^2, \label{eq:flow regime 2}
\end{align}
where $\sigma_k^2 = n p_k (1-p_k)$ is the variance of $H_B$ in the state $|n,p_k\>$ for $k= i,f$ and $\sigma^2 = p(1-p)$ is the variance for a Bernoulli distribution. Note that the variance evaluated for pure states is a genuine measure of coherence~\cite{VarianceAsymmetry} and that due to microscopic energy conservation, that is the fact $U$ commutes with $H_{SB}$, both \textit{energy and variance in energy} are globally conserved. Given this,  Eq.~\eqref{eq:flow regime 1} and Eq.\eqref{eq:flow regime 2} characterise the change in energy and coherence of the system due to an equal and opposite change in the battery. 

Furthermore, binomial states exhibit sub-Poissonian statistics, that is the variance $n p (1-p) $ is smaller than the mean $np$ (for non vanishing $p$). Therefore, it follows from Eq.~\eqref{eq:flow regime 1} and Eq.\eqref{eq:flow regime 2} that the fluctuation relation (\ref{eq:GlobalCrooks}) captures the sub-Poissonian character of these states and shows that this affects the resulting irreversibility of the dynamics. Viewed through the lens of  quantum optics, binomial states of light are anti-bunched~\cite{BinomQuasiProb}, a  signature of non-classicality. Thus the binomial state Crooks equality draws a non-trivial link between bunching and the reversibility of quantum driving processes, since anti-bunching and sub-Poissonian statistics are directly correlated for single-mode time-independent fields~\cite{BinomialStatesStoler}.

In the case of spin-coherent states, the Hamiltonian is in effect taken to be defined in the eigenbasis of the spin-$z$ operator and therefore the variances in Eqs (\ref{eq:flow regime 1}) and (\ref{eq:flow regime 2}) detail the variation of uncertainty in the spin-$z$ component.
However, aligning the Hamiltonians in the $z$-direction defines a preferential axis and therefore the spin-$z$ and the spin-$x$ and spin-$y$   components are not placed on equal footing. This is because the effective potential is invariant under unitary transformations $U$ that commute with $H$, that is
\begin{equation}
    \tilde{E} (\beta, H, \rho) =  \tilde{E} (\beta, H, U \rho U^\dagger) \  \ \ \forall \ [U, H]=0 \; ,
\end{equation}
and hence is invariant under rotations about the $z$-axis. Consequently, while the fluctuation relation captures changes to the uncertainties in the spin-$z$ components, the relation is unaffected by changes to uncertainties in the spin-$x$ and spin-$y$ components. More generally, the invariance of the effective potential to phase rotations means that even for standard coherent states, the fluctuation relation depends on the magnitude of the absolute displacement but not the particular magnitude of the expectation values for position and momentum. This is no coincidence, as the connection between these regimes will be explored further on.

\bigskip

\textit{Deviations from Classicality. $\ \ $} To characterise the  deviations between the binomial state Crooks relation and classical Crooks equality, we can compare the generalised energy flow $\Delta \tilde{W}$ to the actual energy flow in the forwards and reverse processes. In the standard Crooks equality, the work term appearing in the exponent of (\ref{eq:Crooks}) can be expressed as  $W =  (W - (-W))/2$,  the average difference between the work done in the forward and reverse processes. For the quantum analogue, we introduce 
\begin{equation}
    W_q = (\Delta E_+ - \Delta E_-)/2
\end{equation}
as the difference between the energy cost $\Delta E_+$ of the forwards process
and the energy gain $\Delta E_-$ of the reverse process.
Restricted to binomial state preparations of the form (\ref{eq:bin operators}), the  binomial states Crooks relation is
\begin{equation}
    \frac{\mathcal{P}(n_f,p_f|\gamma_i; n_i,\tilde{p}_i)}{\mathcal{P}(n_i,p_i|\gamma_f; n_f,\tilde{p}_f)}
    = \exp \Big(\beta \left(q(\chi) W_q-  \Delta F \right)\Big). 
\end{equation}
where the transition probabilities take the form
\begin{equation}
\begin{aligned}
   \P(n_f,p_f| \gamma_i;n_i,\tilde{p}_i) \\ :=\mathcal{Q} \Big(  \I \otimes |n_f,p_f\>\<n_f,p_f| \Big\vert  \gamma_i \otimes &|n_i,\tilde{p}_i\>\< n_i , \tilde{p}_i| \Big)  
\end{aligned}
\end{equation}
and we introduce the quantum distortion factor
\begin{equation}
    q(\chi) := \frac{\Delta \tilde{W}}{W_q}
\end{equation}
as the ratio between the generalised work flow and the actual energy flows. The classical limit $q(\chi) = 1$ corresponds to a quasi-classical expression in which the quantum fluctuation relation depends only on the energy difference between the two states $|n_i,p_i\>$ and $|n_f,p_f\>$. This can be seen from Eqs. (\ref{eq:flow regime 1},\ref{eq:flow regime 2}) when truncating to first order in $\beta \hbar \omega$. Deviations from unity thus capture the quantum features of the process. 

\medskip

The resizing and re-aligning protocols experience two related yet distinct distortions. These factors, derived in Appendix \ref{ap:Binomial}, are
\begin{equation}\label{eq:q1}
\begin{aligned}
&q_{\mbox{\tiny align}} (\chi) = \frac{1}{\chi} \frac{\ln(\tilde{p}_f/p_f) - \ln( \tilde{p}_i / p_i)}{(\tilde{p}_f - \tilde{p}_i ) + (p_f - p_i)}  \ \ \ \text{and} \\  &q_{\mbox{\tiny size}} (\chi) = \frac{1}{\chi} \frac{\ln(\tilde{p}/p) + 2\chi}{\tilde{p} + p}
\end{aligned}
\end{equation}
respectively, again provided neither $p_i$ nor $p_f$ vanish. These two factors are plotted in Fig. \ref{Fig: bin q}. They are equal to each other if one of either $p_i$ or $p_f$ are zero, corresponding to measuring the battery in the ground state, as can be seen with the long-form equations provided in Appendix \ref{ap:Binomial} (see Eqs (\ref{ap eq: q1}) and (\ref{ap eq:q2})). 

Both factors are independent of the system size $n$. 
That only the parameter $p$ plays a non-trivial role is relevant to the fact that it alone controls the coherent properties of binomial states.
 Since $n$ is the free parameter of the resizing regime, it is particularly significant that
 the deviation is independent of the change in system size. Beyond this, the realignment factor  is symmetric in the parameters $p_i , \, p_f$ and thus does not depend on the chosen ordering of the measurements (likewise for the resizing factor with respect to $n_i$ and $n_f$).

\begin{figure}[t]
\centering
\subfloat{\includegraphics[width=\linewidth]{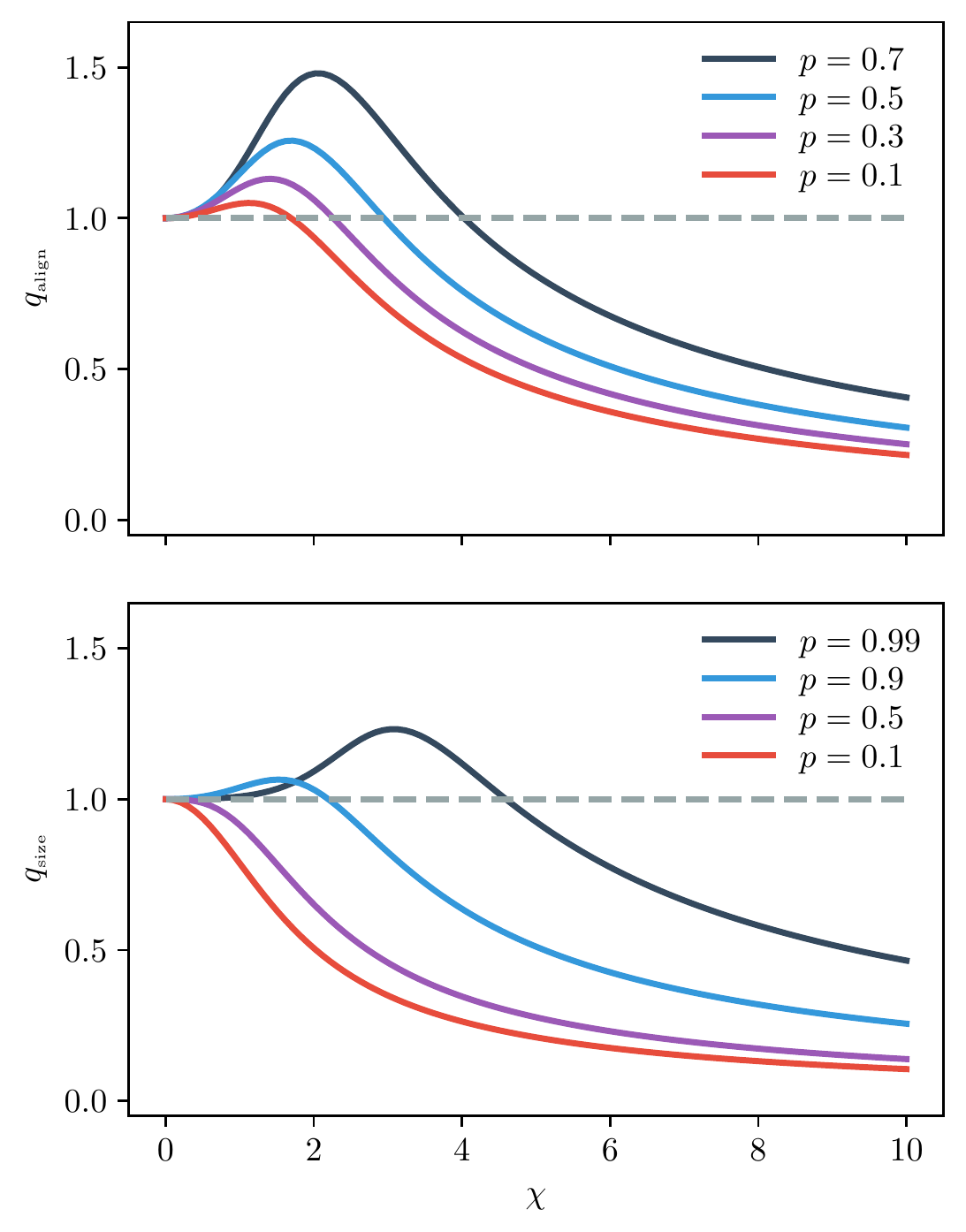}}
\caption{\textbf{Quantum distortions of fluctuation relations due to binomial battery states:} The upper and lower plots correspond to $q_{\mbox{\tiny align}}$ and $q_{\mbox{\tiny size}}$ respectively. The left plot is evaluated for a fixed value $p_f =0.8$. Both functions are plotted against the quantum-thermodynamic ratio $\chi = \frac{\beta \hbar \omega}{2 }$. The plots show that the distortion due to quantum features can both enhance and suppress irreversibility in a process as compared to a `classical equivalent' solely involving energy exchanges. In both cases, we typically find suppressed irreversibility as quantum features dominate for large values of $\chi$. Yet when thermodynamic and quantum energy scales are of similar magnitude, we observe unexpected behaviour.}
\label{Fig: bin q}
\end{figure}

Regarding the thermodynamic properties, both factors exhibit a sensible classical limit in the thermally dominated regime where $\chi$ is much less than one and $\tilde{p} $ converges to $p$. More generally, in the quantum dominated regime at large $\chi$, the distortion is generally sub-unity scaling as $1/\chi$, showing the irreversibility is milder than is classically expected. To understand this, consider the fact that $\tilde{E}(\beta,H,\rho)$ is lower bounded by  $E^{\rm min}(\rho)$, defined as the smallest energy eigenvalue with non-zero weight in the state  $\rho$~\cite{erick}, corresponding to the vacuum energy for all binomial states with $p<1$.  In the low temperature limit, the lower bound is saturated meaning that the generalised energy flow (accounted for by the differences in $\tilde{E}$ between any two states) vanishes. Yet, as shown in Fig. \ref{Fig: bin q}, this behaviour is not true for all temperatures and  values of $p$. 

In the resizing regime, for values of $p$ nearing unity, there exists a finite temperature region where the fluctuation relation exhibits stronger-than-classical irreversibility. 
Peaking for values of $p \approx 1$ in the intermediate region originate because the semi-classical two-point measurement scheme is recovered when $p=1$, which corresponds to an energy eigenstate, hence $q_{\mbox{\tiny size}}(\chi) = 1$. The states satisfying this condition on $p$ must remain close to an energy eigenstate and have a flatter initial slope until larger values of $\chi$ overcome this almost-eigenstate behaviour and recover the $1/\chi$ scaling. 

The behaviour of the realigning regime is more nuanced, having two free parameters. We observe greatly enhanced irreversibility over a finite temperature region for most values of $p_i$ or $p_f$ if the sum of these values are $\gtrapprox 1$. 
An oddity occurs when one measures an excited energy eigenstate, corresponding to\footnote{Due to symmetry in the parameters, one can also set $p_i = 1$ and let $p_f$ be free.} $p_f=1$. In this case, at extremely low temperatures Eq. (\ref{eq:q1}) is modified to
\begin{equation}
\lim_{\chi \to \infty} q_{\mbox{\tiny align}} (\chi) =  \frac{2}{2 - p_i} \geq  1,
\end{equation}
and the quantum regime no longer asymptotically approaches zero. Rather, we have that  $\tilde{E}(\beta,H,\rho)$ is naturally upper bounded by $E^{\rm max}(\rho)$, defined as the largest eigenvalue with non-zero weight in the state $\rho$~\cite{erick}. With the battery prepared in the excited state for either the forward or reverse protocol, we have that $\tilde{E}(\beta, H_B,|n,1\>) = E^{\rm max}$, and  the greatest possible generalised energy flow of $\Delta \tilde{W} = E^{\rm max} - E^{\rm min}$ occurs when the lower bound of $E^{\rm min}$ is saturated.
By fixing one state to be the excited energy eigenstate, the generalised energy flow only attains this upper bound when the temperature reaches absolute zero.

At low temperatures, for values of $p$ nearing unity the deviations from classicality are most pronounced for both regimes. Due to the temperature-dependent rescaling, this choice of parameter  corresponds to the physical preparation of  states with greater coherence present, as detailed by Eq. (\ref{eq:bin operators}) where $p $ is greater than $\tilde{p}$ for all positive temperatures. Initialising the battery in a state with a large amount of coherence thus generates the non-classical behaviour we would expect.

\medskip

From this analysis, we can conclude that binomial states batteries display a greater range of distinguishing features than coherent states, with the coherent properties playing a highly non-trivial role. We observe behaviour that is reminiscent of the semi-classical coherent state Crooks equality in the high and low temperature limits. In an intermediate temperature region however, we observe deviations that lead to stronger than classical irreversibility in both the resizing and realignment regimes.
We note that the binomial state factors bear many qualitative similarities to the squeezed-state factors derived in~\cite{CoherentFluct}. The connection between binomial and coherent states in an appropriate limit are discussed next.

\bigskip

\textit{The Harmonic Limit. $\ \ $}
Infinite dimensional binomial states in harmonic systems exhibit behaviour that approaches simple harmonic motion. This link is well established and leads to a semi-classical limit for the binomial state fluctuation theorem. 
Specifically, as shown in Appendix~\ref{ap:Binomial}, we prove that as $n$ tends to infinity, the binomial state $|n,p\> $ tends to the coherent state $|\alpha\>$ where the displacement parameter is given by $\alpha =\sqrt{np}$ and thus is only defined as long as $np$ remains finite.
Consequently, for infinitely large spin systems, or infinitely large ensembles of qubits, with a finite expected polarisation, binomial states reduce to coherent states. Thus, in this limit, the binomial state and coherent state Crooks equalities~\cite{CoherentFluct} are quantitatively and qualitatively identical.

It follows that for infinite dimensional binomial states $q_{\mbox{\tiny align}} (\chi)$ and $q_{\mbox{\tiny size}}(\chi)$ converge on
\begin{equation}\label{eq:qCoherentState}
    q(\chi) = \frac{1}{\chi} \tanh (\chi)
\end{equation} 
This form admits a special interpretation in terms of the mean energy of a harmonic oscillator $\hbar \omega_{\rm th} := \< H_B \>_\gamma$, with $q(\chi) = k_B T/\hbar \omega_{\rm th}$. In particular, the average energy in a thermal harmonic oscillator is related to the thermal de Broglie wavelength $\lambda_{\rm th}$~\cite{CoherentFluct}. The thermal de Broglie wavelength often finds use as a heuristic tool to differentiate between quantum and thermodynamic regimes. The coherent state equality thus leads to a natural and smooth transition between quantum and thermal properties for semi-classical battery states delineated by $\lambda_{\rm th}$, suggesting a genuinely quantum-thermodynamic relation.


It is interesting then that the binomial state fluctuation relation is able to incorporate a wide-ranging set of features, all the way from the highly quantum single-qubit states to the semi-classical coherent state limit, together in the same framework. 


\subsection{Energy translation invariance, Jarzynski relations and stochastic entropy production}\label{sec:Jarz}

The photon added and subtracted Crooks equalities both quantify transition probabilities between states of the battery. If we assume that the system and battery dynamics depend only on the change in energy of the battery and not the initial energy of the battery then we can rewrite the relation in terms of the probability distributions for the change in energy of the battery, that is the work done on the system. This conceptual move allows us to derive a Jarzyski-like relation for photon added and subtracted thermal states and hint at a link between the generalised free energy change and stochastic entropy production.

If the system and battery dynamics are independent of the initial energy of the battery, then the following energy translation invariance condition holds
\begin{equation}
  \P(E_j | \gamma_i^\pm, E_k) =  \P((E_j - E_k) + E_l | \gamma_i^\pm, E_l)  \ \ \ \forall \ E_j,\, E_k, E_l\, \; .
\end{equation}
We can now define the work probability distributions in the forwards (F) and reverse (R) processes for the photon added (+) and subtracted relations (-) as 
\begin{align}
    \P_F^\pm(W) &:= \sum_{w} \P \left(E_0 - w | \gamma_i^\pm, E_0\right)  p\left(E_0\right) \delta\left(W - w  \right)  \ \ \ \text{and} \\
    \P_R^\pm(W) &:= \sum_{w} \P\left(E_0 - w | \gamma_f^\pm, E_0\right)  p\left(E_0\right) \delta\left(W - w  \right) \ 
\end{align}
where $p(E_0)$ is the probability that the battery is prepared with energy $E_0$. It now follows, as shown in Appendix~\ref{ap:DerivationPhotonAddSub} that the photon added and subtracted Crooks relation can be written explicitly in terms of these work distributions as 
\begin{equation}
      \frac{\P_F^\pm(W)}{\P_R^\pm(-W)} = \R_\pm(W) \exp\left(\beta (W \mp \Delta E_{\mbox{\footnotesize vac}}- 2\Delta F )\right) \; . 
\end{equation}

The classical Jarzynski equality, which quantifies the work done by a driven system for a \textit{single} driving process, emerges as a corollary to the classical Crooks equality. Similarly, here by rearranging and taking the expectation of both sides of the above equality we obtain the photon added and subtracted Jarzynski relation
\begin{equation}
        \left\langle  \frac{1}{\R_\pm(W)} \exp(-\beta W) \right\rangle =  \exp \left(-\beta (2 \Delta F \pm \Delta E_{\mbox{\footnotesize vac}} )\right)  \; .
\end{equation}
This relation complements our Crooks relation, Eq.~\eqref{eq:addsubCrooks}, by relating the work done on the athermal system for a \textit{single} driving process, where the system's Hamiltonian is changed from $H_S^i$ to $H_S^f$, to the associated change in free energy. 

In classical stochastic thermodynamics, when generalising fluctuation relations to non-equilibrium initial states, such as photon added or subtracted thermal states, a natural quantity to consider is the stochastic entropy production. As expected and as shown in \cite{hyukjoon}, in the limit of a classical battery which is assumed to be energy translation invariant, this inclusive setting obeys the classical Crooks equality in its formulation in terms of stochastic entropy production~\cite{Crooks}. This suggests it may be possible to directly relate the generalised free energies term of the global fluctuation relation for non-equilibrium system states to stochastic entropy production. While these ideas were touched on in \cite{hyukjoon}, explicitly stating this link remains an open question.

\medskip

An analogous approach for the binomial state Crooks equality encounters difficulties. States with coherence undergo a temperature dependent rescaling and therefore the initial and final states in the forwards and reverse process are related but not equivalent. Thus due to the presence of coherence, energy translation invariance is not a sufficient condition to rewrite the binomial state Crooks relation in terms of work probability distribution. Therefore we cannot derive a Jarzysnki-like equality and the link with stochastic entropy production is further obscured. Similar problems arise for states such as coherent, squeezed and Schr\"odinger cat states as were studied in \cite{CoherentFluct}.

\section{Conclusions and Outlook}
In this paper we have probed deviations from the classical Crooks equality induced by the initial state of the system or battery and the measurements made at the end of the driving process. However, we stress that the choice in prepared states and measurement operators is not the only manner in which the relation is non-classical. Rather the dynamics induced by the unitary evolution will in general entangle the system and battery resulting in coherence being exchanged between the two systems. Thus the evolved state may be a highly non-classical state. For example, for the coherent state Crooks equality the battery are prepared in a coherent state, the most classical of the motional states of a harmonic oscillator. However, driving the battery with a change in Hamiltonian $H_S^i$ to $H_S^f$ using the experimental scheme proposed in~\cite{CoherentFluct}, results in the highly non-Gaussian state with a substantially negative Wigner function. The non-classicality of the final state can be amplified by repeating the driving process a number of times, that is cycling through changes of $H_S^i$ to $H_S^f$ back to $H_S^i$ and again to $H_S^f$ repeatedly. 

\begin{figure}[t]
  \centering
{\includegraphics[width=\linewidth]{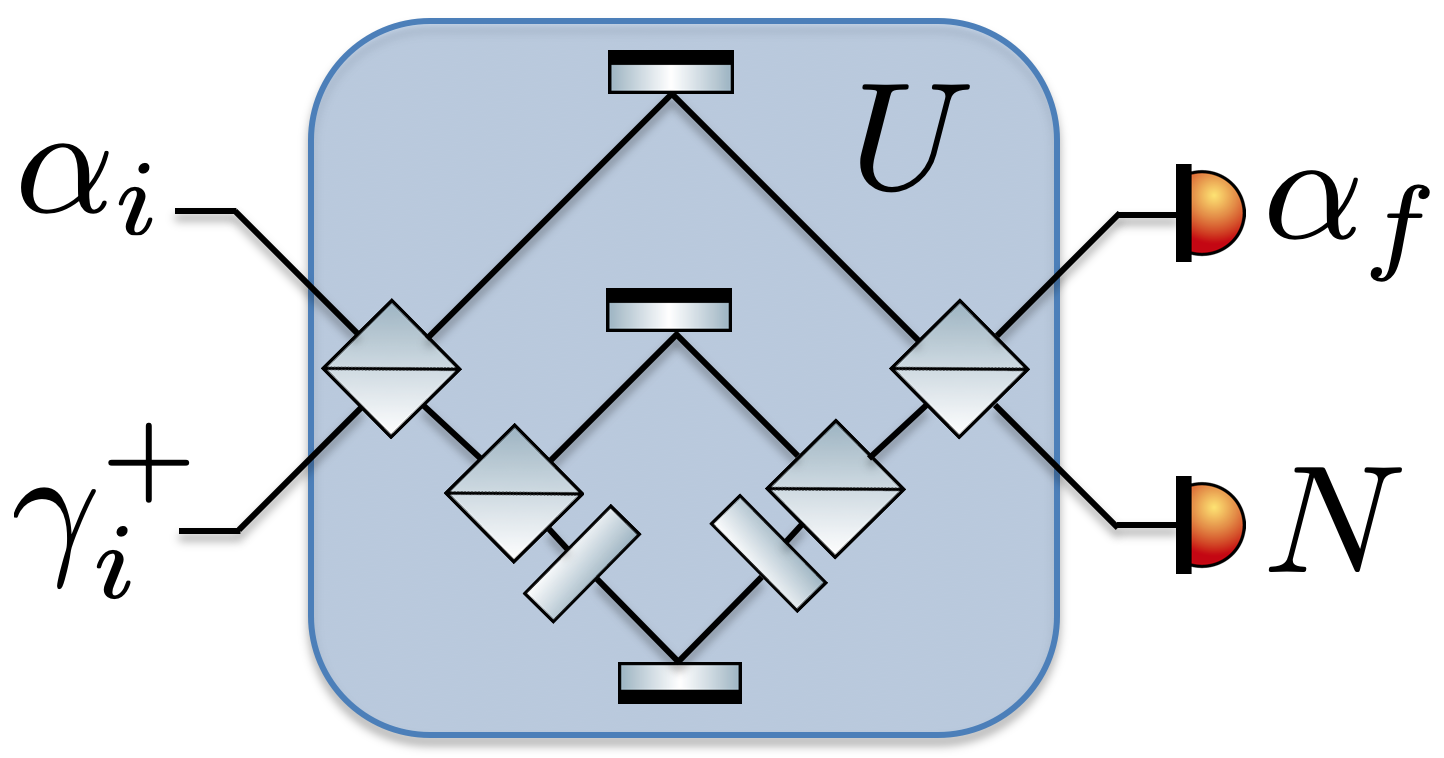}}
\caption{\textbf{Linear optic implementation schematic.} A photon added (or subtracted) thermal state is sent into one input arm of a linear optical set up and a coherent state the other. The linear optical set up, consisting of a series of linear optical elements such as beamsplitters, phase-shifters and mirrors (the particular sequence sketched here is chosen  arbitrarily), drives the photonic system and battery with an energy conserving and time reversal invariant operation. Finally, a coherent state measurement is performed on one output arm of the optical setup using a homodyne detection and the number of photons out put is measured in the other arm. }
\label{fig:ExpSchematic}
\end{figure}

The photon added and subtracted Crooks relations could be tested by supposing that both the system and battery are photonic and using a linear optical setup, as sketched in Fig.~\ref{fig:ExpSchematic}. Preparing a photonic battery in a high energy eigenstate, that is a Fock state containing a large but well defined number of particles, would be experimentally challenging and thus a more promising avenue is to consider a battery in a coherent state by driving one input arm with a laser. Such a scenario would be quantified by a coherent state photon added and subtracted Crooks relation.
A limitation of this implementation is that it would not change the effective Hamiltonian of the system and thus only probe the relation in the limit that $\Delta F$ and $\Delta  E_{\mbox{\footnotesize vac}}$ vanish. Constructing a physical implementation involving a change to the system frequency requires more imagination. One possibility would be to generalise the trapped ion implementation proposed in~\cite{CoherentFluct} but use a pair of internal energy levels to simulate a thermal state of an oscillator. This could be done by changing the background potential to simulate a wider range of energy level splittings.

One possible means of testing the binomial state Crooks equality would be to prepare a finite number of qubits in the state $|p\> = \sqrt{1-p}|0\> + \sqrt{p}|1\>$ and perform a unitary algorithm that interacts the qubits with a thermal system. This could perhaps be best performed on a quantum computer by utilising methods for Hamiltonian simulation~\cite{VQE, TimeDepHamSim} and with the thermal system modelled using `pre-processing'~\cite{CoherentFluct}. One would need to restrict to unitaries that conserve energy between the qubits and the thermal system. Both regimes could be probed with this set-up, where one could have an $N$ qubit register and in one case prepare $n_i$ or $n_f$ qubits in the state $|p\>$, where $n_i,n_f \leq N$, or in the other case a fixed number of qubits could be individually addressed to rotate them in the Bloch sphere.  Measurements in different bases are routinely performed on quantum computers and thus the measurement procedure is readily implemented.

We have taken a highly general but rather abstract fluctuation relation and shown how its physical content can be elucidated through a study of particular examples of interest. However, the cases we have considered are just a sample of the diverse range of phenomena that can be explored with this framework. While we have developed Crooks equalities for thermal systems to which a single photon has been added or subtracted, a natural extension to probe further perturbations from thermality would be to generalise our results to the case where multiple photons are added to or subtracted from the thermal state, or perhaps the case when a photon is added and then subtracted from a thermal state. Similarly, one could quantify higher order quantum corrections to the Crooks relation by developing equalities for squeezed and cat binomial states. On a different note, incoherent binomial states, that is the dephased variant of a binomial state, model Fock states that have been transmitted through a lossy channel and thus model a lossy classical battery. Given the structural similarities between incoherent and coherent binomial states, our results here could be used to develop Crooks relations for imperfect batteries.


\acknowledgements
We are grateful for insightful conversations with Hyukjoon Kwon. We acknowledge support from the Engineering and Physical Sciences Research Council Centre for Doctoral Training in Controlled Quantum Dynamics and the Engineering and Physical Sciences Research Council. 

\bibliographystyle{ieeetr} 
\bibliography{refs}

\appendix

\onecolumngrid

\appendix

\section{Derivation of photon added and subtracted Crooks equality}\label{ap:DerivationPhotonAddSub}

The photon added (subtracted) Crooks equality is derived from the global fluctuation relation by supposing that the system is prepared in a photon added (subtracted) thermal state. That is for the photon added (+) and photon subtracted (-) equalities we suppose that the system is prepared in the states 
\begin{align}
     &\rho_S^i = \gamma_{H_S^i}^\pm  \ \ \ \text{and} \ \ \ \rho_S^f =  \gamma_{H_S^f}^\pm
\end{align}
for the forwards and reverse process respectively where the photon added state and subtracted states are defined as
\begin{equation}
    \gamma_{H}^+ \propto a^\dagger \exp\left(- \beta H \right) a  \ \ \ \text{and} \ \ \  \gamma_{H}^- \propto a \exp\left(- \beta H \right) a^\dagger \; 
\end{equation}
respectively.
In what follows we will use the short hand $\gamma_{i}^\pm \equiv \gamma_{H_S^i}^\pm$ and $\gamma_{f}^\pm \equiv \gamma_{H_S^f}^\pm$ to simplify notation. For concreteness we consider an quantum harmonic oscillator system with initial and final Hamiltonians given by 
\begin{equation}\label{eq:SystemHamiltonians}
    H_S^k := \hbar  \omega_k \left(a_k^\dagger a_k  + \frac{1}{2}\right) \; , 
\end{equation}
for $k =i$ and $k= f$, such that the system is driven by a change in its frequency from $\omega_i$ to $\omega_f$. 

\medskip

We leave the battery Hamiltonian $H_B$ entirely general and in order to isolate the deviations to the classical Crooks equality due to the athermality of the initial system states, we consider a semi classical battery which is prepared and measured in the energy eigenbasis. Specifically we assume that 
\begin{equation}
    \rho_B^i = \ket{E_i}\bra{E_i} \ \ \ \text{and} \ \ \ \rho_B^f = \ket{E_f}\bra{E_f} \  
\end{equation}
where $\ket{E_i}$ and $\ket{E_f}$ are energy eigenstates of $H_B$. Given that the battery is prepared in energy eigenstates, the measurement operators $X_B^i$ and $X_B^f$ specified by Eq.~\eqref{eq:Mapping} are also projectors onto energy eigenstates, that is 
\begin{equation}
    X_B^i = \ket{E_i}\bra{E_i} \ \ \ \text{and} \ \ \ X_B^f = \ket{E_f}\bra{E_f} \; .
\end{equation}
It follows that the generalised energy flow $\Delta \tilde{E}$, Eq.~\eqref{eq:GenWork}, evaluates to the change in energy of the battery,
\begin{equation}
    \Delta \tilde{W} = E_i - E_f \equiv W \; , 
\end{equation}
which by global energy conservation is equivalent to the work done, $W$, on the system. 

\medskip

To derive the photon added and subtracted Crooks relations from the global fluctuation, we need to determine the measurement operators $X_S^i$ and $X_S^f$ which are related to the initial photon added and subtracted states by the mapping $\M$, Eq.~\eqref{eq:Mapping}. Specifically inverting Eq.~\eqref{eq:Mapping} we have that the measurement operators for the photon added, $X_S^{k+}$, and subtracted, $X_S^{k-}$, cases respectively are related to the photon added and subtracted thermal states by
\begin{align}
     X_S^{k\pm} \propto  \exp\left(\chi_k a_k^\dagger a_k \right) \gamma_{k}^\pm \exp\left(\chi_k a_k^\dagger a_k \right) 
\end{align}
where $\chi_k = \frac{\beta \hbar \omega_k }{2 }$. On substituting in the explicit expressions for $\gamma_{k}^+$ and $ \gamma_{k}^-$, and using the Hadamard Lemma, we find that
\begin{align}
      X_S^{k+} &\propto \exp\left(\chi_k a_k^\dagger a_k \right) a_k^\dagger \exp\left(- 2\chi_k a_k^\dagger a_k \right) a_k \exp\left(\chi_k a_k^\dagger a_k \right)  
     \propto a_k^\dagger a_k  \ \ \ \ \ \ \text{and similarly,}  \\
     X_S^{k-} &\propto \exp\left(\chi_k a_k^\dagger a_k \right) a_k\exp\left(-2 \chi_k a_k^\dagger a_k \right) a_k^\dagger  \exp\left(\chi_k a_k^\dagger a_k \right)  
     \propto a_k a_k^\dagger   \; . 
\end{align}
We note that any constants of proportionality in front of the measurement operators $X_S^i$ and $X_S^f$ will cancel out in the final relation and thus we are free to set them to 1. We therefore conclude that the measurement operators for the photon added Crooks relation, forced by the mapping $\M$, Eq.~\eqref{eq:Mapping}, are given by
\begin{align}\label{eq:MeasureOps1}
     &X_S^{i+} = a_i^\dagger a_i \equiv N_i \  \ \ \text{and} \ \ \ X_S^{f+} = N_f
\end{align}
and the measurement operators for the photon subtracted equality are equal to
\begin{align}\label{eq:MeasureOps2}
     &X_S^{i-} = a_i a_i^\dagger = N_i + 1 \ \ \text{and} \ \ \ X_S^{f-} = N_f +1 \; 
\end{align}
where $N_i$ and $N_f$ are the initial and final number operators respectively.

The photon added Crooks equality thus quantifies the ratio of \begin{equation}
 \begin{aligned}
    \mathcal{Q}\left(a_f^\dagger a_f \otimes \ket{E_f}\bra{E_f}\bigg{|}\gamma_i^+ \otimes \ket{E_i}\bra{E_i}\right)  = n(E_f | \gamma_i^+, E_i) \P(E_f | \gamma_i^+, E_i) 
\end{aligned}  
\end{equation}
for a forwards process, and
\begin{equation}
 \begin{aligned}
    \mathcal{Q}\left(a_i^\dagger a_i \otimes \ket{E_i}\bra{E_i}\bigg{|} \gamma_f^+ \otimes \ket{E_f}\bra{E_f}\right)  = n(E_i | \gamma_f^+, E_f) \P(E_i | \gamma_f^+, E_f) 
\end{aligned}  
\end{equation}
of a reverse process. Here $n(E_f | \gamma_i^+, E_i)$ ($n(E_i | \gamma_f^+, E_f)$) is the average number of photons measured in the system at the end of the forwards (reverse) process, conditional on the battery being measured to have the energy $E_f$ ($E_i$). Similarly, the photon subtracted Crooks equality quantifies the ratio of 
\begin{equation}
\begin{aligned}
      \mathcal{Q}\left((a_k^\dagger a_k + 1) \otimes \ket{E_k}\bra{E_k}\bigg{|} \gamma_j^+ \otimes \ket{E_j}\bra{E_j}\right) = \left(n(E_k | \gamma_j^+, E_j) +1 \right) \P(E_k | \gamma_j^+, E_j)   
\end{aligned}
\end{equation}
for a forwards process, with $j=i$ and $k=f$, and a reverse process, with $j=f$ and $k=i$. 

\medskip

It remains to calculate the generalised free energy $\Delta \tilde{F}$ for the measurements $X_S^i$ and $X_S^f$ as defined in Eq.~\eqref{eq:MeasureOps1} and Eq.~\eqref{eq:MeasureOps2}. To do so we start by noting that $\Delta \tilde{F}$ can be written as 
\begin{equation}\label{eq:reWrittenGenFree}
    \Delta \tilde{F} = k_B T \ln\left(\frac{\tilde{Z}\left(\beta, H_S^i, X_S^i \right)}{\tilde{Z}\left(\beta, H_S^f, X_S^f \right)} \right) \ \ \ \text{where} \ \ \ \tilde{Z}\left(\beta, H, X \right) := \Tr[\exp(-\beta H) X] \; .
\end{equation}
As our notation suggests, $\tilde{Z}$ is an operator dependent mathematical generalisation of the usual thermodynamic partition function, 
\begin{equation}\label{eq:ExplicitPartitionFunc}
   Z(\beta, H_S^k) := \Tr[\exp(- \beta H_S^k)] \ .
\end{equation}
For the oscillator Hamiltonians defined in Eq.~\eqref{eq:SystemHamiltonians}, we find by working in the number basis that
\begin{equation}
   \begin{aligned}\label{eq:ExplicitFtildeAddSub}
&\tilde{Z}\left(\beta, H_S^k, N_k \right)
=  \sum_{n_k=0}^\infty  n_k  \exp(-2\chi_k (n_k + 1/2))
=  \frac{\exp(\chi_k)}{(\exp(2\chi_k)-1)^2} \ \ \ \ \ \ \text{and} \\
&\tilde{Z}\left(\beta, H_S^k, N_k+1 \right)
=  \sum_{n_k=0}^\infty  (n_k + 1)  \exp(-2\chi_k (n_k + 1/2)) =  \frac{\exp(3\chi_k)}{(\exp(2\chi_k)-1)^2}    \;. 
\end{aligned} 
\end{equation}
The physical content of these expressions can be elucidated by rewriting them in terms of the usual partition function, which evaluates to
\begin{equation}\label{eq:ExplicitPartitionFunc2}
 Z(\beta, H_S^k) = \frac{\exp(\chi_k)}{\exp(2\chi_k)-1} \ .
\end{equation}
On substituting Eq.~\eqref{eq:ExplicitPartitionFunc2} into Eq.~\eqref{eq:ExplicitFtildeAddSub} we obtain
\begin{equation}
    \begin{aligned}\label{eq:GenFreePartFunc}
    &\tilde{Z}\left(\beta, H_S^k, N_k \right) : =  Z_k \frac{1}{\exp(2\chi_k)-1} = (Z_k)^2 \exp(-\chi_k) \ \ \ \text{and} \\
    &\tilde{Z}\left(\beta, H_S^k, N_k + 1 \right) = Z_k \frac{\exp(2\chi_k) }{\exp(2\chi_k)-1} = (Z_k)^2 \exp(\chi_k) 
\end{aligned}
\end{equation}
where we have introduced the short hand $Z_k \equiv  Z(\beta, H_S^k)$. Finally, on substituting Eq.~\eqref{eq:GenFreePartFunc} into Eq.~\eqref{eq:reWrittenGenFree}, and using the fact that because 
\begin{equation}
     \frac{Z_f}{Z_i}  =  \exp(-\Delta F/k_B T) \ \ \ \text{it follows that}  \ \ \  \left(\frac{Z_f}{Z_i}\right)^2  =  \exp(-2 \Delta F/k_B T)
\end{equation}
we find that 
\begin{equation}
    \Delta \tilde{F}^\pm = \pm \Delta E_{\mbox{\footnotesize vac}} + 2 \Delta F \, .
\end{equation}
In the above we have introduced $\Delta E_{\mbox{\footnotesize vac}}$ as the difference between the vacuum energies of the harmonic oscillator at the start and end of the forwards driving process, 
\begin{equation}
    \Delta E_{\mbox{\footnotesize vac}} := \frac{1}{2}\hbar \omega_f - \frac{1}{2} \hbar \omega_i \; .
\end{equation}

\medskip

The photon added (+) and photon subtracted (-) Crooks equality can thus be written as
\begin{equation}\label{eq:APaddsubCrooks}
    \frac{\P(E_f | \gamma_i^\pm, E_i)}{\P(E_i | \gamma_f^\pm, E_f)} = \R_\pm \exp\left(\beta (W \mp \Delta E_{\mbox{\footnotesize vac}}- 2\Delta F )\right) \; ,
\end{equation}
where the prefactors $\R_+$ and $\R_-$ are defined as
\begin{align}\label{eq:Rplusminus}
    &\R_+ := \frac{ n(E_i | \gamma_f^+, E_f)}{ n(E_f | \gamma_i^+, E_i)} \ \ \ \text{and}  \ \ \  \R_- := \frac{ n(E_i | \gamma_f^-, E_f)+1}{ n(E_f | \gamma_i^-, E_i)+1} \ .
\end{align}
Since the number of photons in the system is necessarily a positive quantity, the prefactors are only defined when both the numerator and denominator of Eq.~\eqref{eq:Rplusminus} are positive quantities.


The physical role of the $\R_\pm$ term can be made more explicit by taking advantage of that fact that energy is conserved during the driving process. It follows that the number of photons at the end of the driving process is equal to the average number of photons initially in the system plus (minus) the change in photon number due to the decrease (increase) in the energy of the battery. By energy conservation we can write
\begin{align}\label{eq:ConditionalNumberMeaures}
   \hbar \omega_f \left(n(E_f | \gamma_i^\pm, E_i)+ \frac{1}{2}\right) = \hbar \omega_i \left( n_i^\pm  + \frac{1}{2}\right)  - W  \ \ \ \text{and}  \ \ \ \hbar \omega_i \left(n(E_i | \gamma_f^\pm, E_f) + \frac{1}{2} \right) = \hbar \omega_f \left(n_f^\pm + \frac{1}{2}\right) + W \ \ \ \; 
\end{align}
where $ n_i^\pm$ ($ n_f^\pm$) is the average number of photons in a photon added/subtracted thermal state with frequency $\omega_i$ ($\omega_f$) at temperature $T$. 
Eq.~\ref{eq:ConditionalNumberMeaures} can be rearranged to find the average number of photons measured at the end of the driving processes, 
\begin{align}\label{eq:ConditionalNumberMeauresRearrange}
   \hbar \omega_f n(E_f | \gamma_i^\pm, E_i) = \hbar \omega_i n_i^\pm  - W -  \Delta E_{\mbox{\tiny vac}} \ \ \ \text{and}  \ \ \ \hbar \omega_i n(E_i | \gamma_f^\pm, E_f)  = \hbar \omega_f n_f^\pm  + W  +\Delta E_{\mbox{\tiny vac}}  \ \ \ \; .
\end{align}
Thus, on substituting Eq.~\eqref{eq:ConditionalNumberMeauresRearrange} in Eq.~\eqref{eq:Rplusminus} we find that the prefactor $\R_\pm$ takes the form
\begin{equation}\label{eq:plusminusR1}
    \begin{aligned}
        &\R_\pm(W) = \frac{\omega_f}{\omega_i} \frac{\hbar \omega_f n_f^\pm +  W + x_\pm}{\hbar \omega_i n_i^\pm - W \mp x_\pm} 
\end{aligned}
\end{equation}
with $x_+$ equal to the \textit{change} in vacuum energy, $x_+= \Delta E_{\mbox{\tiny vac}}$, and $x_-$ equal to the sum of the initial and final vacuum energies $x_- = \frac{\hbar \omega_f + \hbar \omega_i}{2}$. As discussed in Section~\ref{sec:PhotonAddSub} of the main text, the average number of photons in a photon added or subtracted state,  $ n_f^\pm$, evaluates to
\begin{equation}
    n_k^+ = 2 \bar{n}_k + 1  \ \ \ \text{and} \ \ \ n_k^- = 2 \bar{n}_k \ 
\end{equation}
where $\bar{n}_k$ is the average number of photons in a thermal state with frequency $\omega_k$ and takes the form
\begin{equation}
    \bar{n}_k := \frac{1}{Z_k} \sum n_k \exp(- 2 \chi_k (n_k + 1/2)) = \frac{1}{\exp(2\chi_k) -1} \; .
\end{equation}
Thus we find that the prefactor $\R_\pm$, Eq.~\eqref{eq:plusminusR1}, can be rewritten in terms of the mean number of photons in a thermal state as 
\begin{equation}
    \begin{aligned}\label{eq:Prefactors2}
        &\R_\pm(W) = \frac{\omega_f}{\omega_i} \frac{\hbar \omega_f(2 \bar{n}_f + k_\pm) +  W + \Delta E_{\mbox{\tiny vac}}}{\hbar \omega_i \left(2 \bar{n}_i + k_\pm^{-1}\right) - W - \Delta E_{\mbox{\tiny vac}}} 
\end{aligned}
\end{equation}
with $k_+ = 1$ and $k_- = \frac{\omega_i}{\omega_f}$. It is worth noting that the prefactor implicitly depends on the free energy of the initial and final Hamiltonians because the term $\hbar \omega_k (\bar{n}_k + \frac{1}{2})$ is the average energy of a photon in a thermal state with frequency $\omega_k$, which is equal to the free energy of the state plus $k_B T$ times the entropy of the state. Thus $\R$ depends on the temperature, the work done during the driving process, as well as the equilibrium free energy and the entropy of a thermal system with respect to the initial and final Hamiltonians. 

\bigskip

\paragraph*{Photon added and subtracted Jarzynski equality.} 

We can derive a Jarzyski-like relation for photon added and subtracted thermal states from Eq.~\eqref{eq:addsubCrooks}, if we further assume that the system and battery dynamics depend only on the change in energy of the battery and not the initial energy of the battery. That is if the following energy translation invariance condition holds
\begin{equation}
   \P(E_j | \gamma_i^\pm, E_k) =  \P((E_j - E_k) + E_l | \gamma_i^\pm, E_l)  \ \ \ \forall \ E_j,\, E_k, E_l\, \; .
\end{equation}
Having made this assumption we can rewrite the photon added and subtracted Crooks relation, Eq.~\eqref{eq:APaddsubCrooks} as,
\begin{equation}
        \frac{\P(w + E_0| \gamma_i^\pm, E_0)}{\P(- w + E_0| \gamma_f^\pm, E_0)}  = \R_\pm(w) \exp\left(\beta (w \mp \Delta E_{\mbox{\footnotesize vac}}- 2\Delta F )\right) \; ,
\end{equation}
which can be rearranged into 
\begin{equation}\label{eq:RearrangedCrooksAddSub}
       \frac{1}{\R_\pm(w)} \exp(-\beta w) \P(w + E_0| \gamma_i^\pm, E_0) p(E_0)  =  \exp\left(\beta(\mp \Delta E_{\mbox{\footnotesize vac}}- 2\Delta F )\right) \P(- w + E_0| \gamma_f^\pm, E_0) p(E_0) \,  \; 
\end{equation}
where $p(E_0)$ is the probability that the battery is prepared with energy $E_0$. We can now define the work probability distributions in the forwards (F) and reverse (R) processes for the photon added (+) and subtracted relations (-) as 
\begin{align}
    \P_F^\pm(W) &:= \sum_{w} \P \left( E_0-w  | \gamma_i^\pm, E_0\right)  p\left(E_0\right) \delta\left(W - w  \right)  \ \ \ \text{and} \\
    \P_R^\pm(W) &:= \sum_{w} \P\left( E_0-w  | \gamma_f^\pm, E_0\right)  p\left(E_0\right) \delta\left(W - w  \right) \ .
\end{align}
It therefore follows from Eq.~\eqref{eq:RearrangedCrooksAddSub} that the photon added and subtracted Crooks equalities can be rewritten in terms of the forwards and reverse work probability distributions instead of battery state transition probabilities, with
\begin{equation}
    \frac{\P_F^\pm(W)}{\P_R^\pm(-W)} = \R_\pm(W) \exp\left(\beta (W \mp \Delta E_{\mbox{\footnotesize vac}}- 2\Delta F )\right) \; .
\end{equation}
Finally, rearranging and taking the expectation of both sides of the above equality we obtain the photon added and subtracted Jarzynski relation
\begin{equation}
        \left\langle  \frac{1}{\R_\pm(W)} \exp(-\beta W) \right\rangle =  \exp \left(-\beta (2 \Delta F \pm \Delta E_{\mbox{\footnotesize vac}} )\right)  \; .
\end{equation}
Thus we can relate the work done on a system which is prepared in a photon added or subtracted thermal state and driven by a change in Hamiltonian to the change in free energy associated with the change in Hamiltonian.

\section{Derivation of Binomial State Properties}\label{ap:Binomial}

This section contains derivations of some mathematical properties of binomial states. In the main text, the mapping $\M$ was introduced, defined as
\begin{equation}
    \M (X) = \frac{\T \left( e^{-\frac{\beta H_B}{2} } X e^{-\frac{\beta H_B}{2} } \right) }{\Tr(e^{-\beta H} X)}.
\end{equation}
This mapping, sans the time-reversal, is often referred to as a Gibbs rescaling, and it has many interesting properties~\cite{aberg,erick}. Under the Gibbs rescaling, we find the binomial states transform as follows.

\begin{propos}
Let $|n,p\>$ be a binomial state as defined in the main text, for $n \in \mathbb{N}$ and $0 \leq p \leq 1$. For a harmonic Hamiltonian $H_B = \hbar \omega (a^\dagger a + \frac{1}{2})$, the Gibbs re-scaled state $|n,\tilde{p} \> \< n,\tilde{p}| = \Gamma_{H_B}  (|n,p\> \< n,p|) $ is also a binomial state with probability distribution
\begin{equation}\label{eq:rescaled bin}
    \tilde{p} := \frac{e^{-\beta \hbar \omega}p}{p e^{-\beta \hbar \omega} + q}, \ \tilde{q} := \frac{q}{p e^{-\beta \hbar \omega} + q},
\end{equation}
where $q = 1-p$.
\end{propos}
\begin{proof}
Since the Gibbs re-scaling maps pure states to pure states, we need only consider the action of $\mathcal{Z}_{n,p}^{-1/2} e^{-\beta H_B /2} |n,p\> = |\psi\>$ where $\mathcal{Z}_{n,p}^{-1/2}$ is the normalising factor. As the phases are arbitrary,  we neglect them with no loss of generality. Before proceeding, we make the substitution $\chi = \frac{\beta \hbar \omega}{2}$ and $q = 1-p$.
Using the definition of $|n,p\>$, we find
\begin{align}
|\psi\>  &= \frac{1}{\sqrt{\mathcal{Z}_{n,p}}}\sum_{k=0}^n \sqrt{{n \choose k} p^{k} q^{n-k}} e^{ - \chi (a^\dagger a + \frac{1}{2})} |k\>
\\
  &= \frac{1}{\sqrt{\mathcal{Z}_{n,p}}}\sum_{k=0}^n \sqrt{{n \choose k} p^{k} q^{n-k}} e^{ -\chi ( k + \frac{1}{2})} |k\>. \label{eq:rescaling bin}
\end{align}
Let us calculate the normalisation factor
\begin{align}
\mathcal{Z}_{n,p} &= \< n,p | e^{-\beta H_S } |n,p\> \\
             &= \sum_{k = 0}^n \frac{n!}{k!(n-k)!} p^{k} q^{n-k} e^{-2 \chi (k + \frac{1}{2})} \\
             &=  e^{-\chi} (p e^{-2\chi} + q)^n
\end{align}
where to obtain the last line we used the binomial expansion theorem. Inserting this into (\ref{eq:rescaling bin}), we obtain
\begin{align}
   |\psi \>  &= \frac{e^{- \chi/2}}{e^{-\chi/2}(p e^{-2\chi} + q)^{n/2}} \sum_{k=0}^n \sqrt{ {n \choose k} p^k q^{n-k} } e^{-k\chi } |k\> \\
   &=  \sum_{k=0}^n \sqrt{ {n \choose k} \frac{p^k q^{n-k}}{(p e^{-2\chi} + q)^n} e^{-2k\chi }}  |k\> \\
   &= \sum_{k=0}^n \sqrt{ {n \choose k} \left[\frac{p e^{-2\chi}}{p e^{-2\chi} + q} \right]^k \left[\frac{q }{p e^{-2\chi} + q} \right]^{n-k} }  |k\> \\
   &= \sum_{k=0} \sqrt{ {n \choose k} \tilde{p}^k \tilde{q}^{n-k}  } |k\>,
\end{align}
where 
\begin{equation}
    \tilde{p} := \frac{e^{-\beta \hbar \omega}p}{p e^{-\beta \hbar \omega} + q}, \ \tilde{q} := \frac{q}{p e^{-\beta \hbar \omega} + q}.
\end{equation}
It is easily verified that $\tilde{p} + \tilde{q} =1$ and therefore $|\psi \> = |n, \tilde{p}\>$ is a binomial state as claimed. \\
\end{proof}

Binomial state statistics are preserved under a Gibbs re-scaling but in general $\tilde{p}$ decreases with increasing $\chi$, as can be seen if we instead look at $\tilde{q}$. In the limit $\chi \to 0$, $\tilde{q} \to q$ and hence $\tilde{p} \to p$, while in the limit $\chi \to \infty$, $\tilde{q} \to 1$ and conversely $\tilde{p} \to 0$. It smoothly varies between these two limits, implying $\tilde{q} \geq q$.
\\
To derive the quantum distortion factors, we need to know the expectation value in energy for a system prepared in a binomial state.

\begin{propos}\label{prop: Energy Bin}
Suppose $B$ has a harmonic Hamiltonian $H_B := \hbar \omega (a^\dagger a + \frac{1}{2}) $, then the expectation value of energy for a state $|n,p\>$ is 
\begin{equation}
\< H_B \>_{n,p} = \hbar \omega \left(n p + \frac{1}{2} \right).
\end{equation}
\end{propos}
\begin{proof}
We begin by assuming a harmonic Hamiltonian $H_B = \hbar \omega (a^\dagger a + \frac{1}{2})$. Using the definition of $|n,p\>$ leads to
\begin{align}
\< H_B \>_{n,p} &=  \sum_{k = 0}^n \frac{n!}{k!(n-k)!} p^{k} q^{n-k} \hbar \omega  \left(k + \frac{1}{2} \right)    
\end{align}
We now proceed to calculate the two components separately, for the first we have
\begin{align}
    \text{first term} &= \hbar \omega n p \sum_{k=0 }^n k  \frac{(n - 1)! }{k!(n-k)!} p^{k -1} q^{n-k}  \\
    &= \hbar \omega n p \sum_{k=1 }^n   \frac{(n - 1)! }{(k-1)!([n-1]-[k-1])!} p^{k -1} q^{[n-1]-[k-1]} \\
    &= \hbar \omega n p \sum_{j=0 }^m   \frac{m! }{j!(m-j)!} p^{j} q^{m-j} \\
    &= \hbar \omega n p (p + q)^m \\
    &= \hbar \omega np
\end{align}
where we made the substitutions $m =n-1$ and $j = k-1$. Doing a similar calculation for the second term,
\begin{align}
    \text{second term} &= \frac{\hbar \omega  }{2} \sum_{k=0 }^n   \frac{n! }{k!(n-k)!} p^{k } q^{n-k}  \\
    &= \frac{\hbar \omega}{2} (p + q)^n \\
    &= \frac{\hbar \omega}{2}. 
\end{align}
Combining these two equations gives the claimed result.
\end{proof}

The final property we need is the effective potential evaluated for an arbitrary binomial state

\begin{propos}\label{prop: Bin effective potential}
For a binomial state $|n,p\>$ and harmonic Hamiltonian $H_B = \hbar \omega (a^\dagger a + \frac{1}{2})$,
\begin{equation}\label{eq: Bin effective}
    \beta \tilde{E}(\beta,H_B ,|n,p\>) = \frac{\beta \hbar \omega}{2} - n  \ln( p e^{-\beta \hbar \omega} + q )
\end{equation}
where $q = 1-p$. 
\end{propos}
\begin{proof}
From the definition of $\tilde{E}$ and $|n,p\>$ we find
\begin{align}
    \beta \tilde{E} (\beta,H_B , |n,p\>) &= -  \ln \left( \sum_{k=0}^n e^{-\beta \hbar \omega (k + \frac{1}{2})} {n \choose k} p^k q^{n-k}  \right)  \\
    &= -  \ln \left(e^{-\beta \hbar \omega/2} \sum_{k=0}^n e^{-\beta \hbar \omega k } {n \choose k} p^k q^{n-k}  \right)  \\
    &=  \frac{\beta \hbar \omega}{2} - \ln \left( \sum_{k=0}^n  {n \choose k} \left[p e^{-\beta \hbar \omega} \right]^k q^{n-k}  \right)  \\
    &= \frac{\beta \hbar \omega}{2} -  \ln \left(  p e^{-\beta \hbar \omega} + q \right)^n \\
    &= \frac{\beta \hbar \omega}{2} - n  \ln \left(  p e^{-\beta \hbar \omega} + q \right),
\end{align}
which concludes the proof.
\end{proof}

Equation (\ref{eq: Bin effective}) can also be formulated in terms of $\tilde{p}$ by noting that $pe^{-\beta \hbar \omega}/\tilde{p}  = (p e^{-\beta \hbar \omega} + q) $. It follows that
\begin{align}
    \beta \tilde{E}_B (\beta,H_B,|n,p\>) &= \frac{\beta \hbar \omega}{2} - n \ln (pe^{-\beta \hbar \omega}/\tilde{p} ) \\
    &= \beta \hbar \omega(n + \frac{1}{2}) + n \ln (\tilde{p}/p), 
\end{align}
on the condition that $p,\tilde{p} > 0$.

\subsection{The quantum distortion factor for Binomial States}

In the main text we discussed a quantum distortion factor $q(\chi)$ that determines how the quantum fluctuation theorem diverges compared to the standard notion of average change in energy of the forwards and reverse processes. Here we derive the explicit formulae for $q(\chi)$. 

We defined two distinct processes when restricting to binomial state preparation and measurement, corresponding to the \textit{resizing} and \textit{re-aligning} regimes. 
In the re-aligning regime, using Proposition \ref{prop: Energy Bin}, the energetic \textit{cost to the battery} in each protocol is
\begin{align}
    \Delta E_+^{\mbox{\tiny(align)}} &:= \< H_B \>_{n,\tilde{p}_i} - \< H_B \>_{n,p_f} = \hbar \omega n(\tilde{p}_i - p_f), \\
    \Delta E_-^{\mbox{\tiny(align)}} &:= \< H_B \>_{n,\tilde{p}_f} - \< H_B \>_{n,p_f} = \hbar \omega n(\tilde{p}_f - p_i). 
\end{align}
The quantity $W_q^{\mbox{\tiny(align)}} := (\Delta E_+^{\mbox{\tiny(align)}} - \Delta E_-^{\mbox{\tiny(align)}})/2$ therefore takes the form
\begin{equation}
    W_q^{\mbox{\tiny (align)}} = \frac{\hbar \omega n }{2} \left( [\tilde{p}_i + p_i] - [\tilde{p}_f + p_f]  \right) =  \frac{\hbar \omega}{2} \left( p_i \left[ \frac{e^{-\beta \hbar \omega}}{p_i e^{-\beta \hbar \omega} + q_i} + 1 \right] - p_f  \left[ \frac{e^{-\beta \hbar \omega}}{p_f e^{-\beta \hbar \omega} + q_f} + 1 \right]  \right).
\end{equation}
Whereas for the generalised energy flow we can use Proposition \ref{prop: Bin effective potential}, which depends solely upon the normal un-rescaled states  $\Delta \tilde{W} = E(\beta,H_B,|n,p_i\> ) - E(\beta,H_B,|n,p_f\> ) $. This turns out to be
\begin{equation}
    \Delta \tilde{W}_{\mbox{\tiny align}} =  - n k_B T \ln \left( \frac{p_i e^{-\beta \hbar \omega} + q_i}{p_f e^{-\beta \hbar \omega} + q_f}\right) = - n k_B T \ln \left( \frac{ \tilde{p}_f p_i}{p_f \tilde{p}_i  }   \right),
\end{equation}
where the latter equality holds provided $p_f,p_i \neq 0$.
The quantum distortion factor where we are free to vary $p$ for fixed $n$ thus takes the form
\begin{align}\label{ap eq: q1}
    q_{\mbox{\tiny align }}(\chi) &= \frac{1 }{\chi} \frac{\ln \left( \frac{p_i e^{-2 \chi} + q_i}{p_f e^{-2 \chi} + q_f}\right)}{(\tilde{p}_f - \tilde{p}_i) + (p_f - p_i)}
    \\
    &= \frac{1 }{\chi} \frac{\ln \left( \frac{\tilde{p}_f  }{p_f}\right) - \ln \left( \frac{\tilde{p}_i  }{p_i}\right)}{(\tilde{p}_f - \tilde{p}_i) + (p_f - p_i)}, \ \text{ if } p_i,p_f \neq 0.
\end{align}

On the other hand, one is also free to vary $n$ and keep $p$ fixed as detailed by the resizing regime. We can define the same quantities, which we now label with a new supercript to differentiate the cases.

\begin{align}
    \Delta E_+^{\mbox{\tiny(size)}} &:= \< H_B \>_{n_i,\tilde{p}} - \< H_B \>_{n_f,p} = \hbar \omega (n_i \tilde{p} - n_f p) , \\
    \Delta E_-^{\mbox{\tiny(size)}} &:= \< H_B \>_{n_f,\tilde{p}} - \< H_B \>_{n_i,p} = \hbar \omega (n_f \tilde{p} - n_i p), 
\end{align}
which implies
\begin{equation}
    W_q^{\mbox{\tiny(size)}} = \frac{\hbar \omega}{2} \left( n_i - n_f \right)(\tilde{p} + p) = \frac{\hbar \omega p}{2} \left( n_i - n_f \right) \left(\frac{e^{-\beta \hbar \omega}}{p e^{-\beta \hbar \omega} + q} + 1 \right).
\end{equation}
Likewise, the generalised energy flow for this process is given by
\begin{equation}
    \Delta \tilde{W}_{\mbox{\tiny size}} = - (n_i-n_f)k_B T  \ln (p e^{-\beta \hbar \omega} + q) =  (n_i-n_f)\left\lbrace k_B T \ln (\tilde{p}/p) + \beta \hbar \omega \right\rbrace.
\end{equation}
The quantum distortion factor for the second regime is thus
\begin{align}\label{ap eq:q2}
    q_{\mbox{\tiny size}} (\chi)  &= \frac{1}{\chi} \frac{\ln (p e^{-2 \chi} + q)}{\tilde{p} + p}
    \\
    &= \frac{1}{\chi} \frac{\ln(\tilde{p}/p) + 2 \chi}{\tilde{p} + p}, \ \text{ if } p \neq 0.
\end{align}

\subsection{The Harmonic Limit}

In this section we prove that there exists a limit in which binomial states become coherent states with arbitrary precision. In what follows, we assume that $np = \lambda$ for some constant $\lambda \in \mathbb{R}$. The correct limit involves making the binomial states a superposition over infinitely many energy eigenstates by taking $n \to \infty$ and correspondingly $p \to 0$.

Firstly, let us consider the effect on the expectation value for energy. We have that
\begin{align}
    \lim_{\substack{n \to \infty \\ np = \lambda}}\< H_B \>_{n,p} &= \lim_{\substack{n \to \infty \\ np = \lambda}} \hbar \omega (np + \frac{1}{2}) \\
    &= \hbar \omega (\lambda + \frac{1}{2})
\end{align}
which we note bears a likeness to the expectation value of energy for a coherent state $|\alpha\>$ where $|\alpha|^2 = \lambda$. Likewise, the effective potential also attains an identical form to that of a coherent state $\tilde{E} (\beta,H_B,|\alpha\>)$ where we once again choose $|\alpha|^2 = \lambda$.

\begin{align}
    \lim_{\substack{n \to \infty \\ np = \lambda}} \beta \tilde{E}_B (\beta,|n,p\>) &= \lim_{\substack{n \to \infty \\ np = \lambda}} \left( \frac{\beta \hbar \omega}{2} - n  \ln \left(1 + \frac{\lambda}{n}[e^{-\beta \hbar \omega} -1]\right) \right) \\
    &= \lim_{\substack{n \to \infty \\ np = \lambda}} \left( \frac{\beta \hbar \omega}{2} - n  \left[ \frac{\lambda}{n} [e^{-\beta \hbar \omega} -1] + \O \left( \frac{1}{n^2} \right) \right] \right) \\
    &= \frac{\beta \hbar \omega}{2} +  \lambda(1-e^{-\beta \hbar \omega}).
\end{align}
 For our purposes, these two quantities being identical to their coherent state counterparts means that the fluctuation theorem in the appropriate limit is indistinguishable from a coherent state fluctuation theorem. However, it is also the case that the states themselves become identical. This is easily verified by using the closely related characteristic functions~\cite{lukacscharacteristic}. Since characteristic functions $\varphi (t)$ uniquely specify a probability distribution, showing equality for all $t$ translates to equality in distribution. Defining the characteristic function $\varphi_\psi (t) := \< \psi | e^{iH_B t} |\psi\>$ we have
\begin{align}
    \varphi_\alpha (t) &= e^{|\alpha|^2 ( e^{i \hbar \omega t} - 1) + i\frac{\hbar \omega }{2} t} \\
    \varphi_{n,p} (t) &= e^{ i\frac{\hbar \omega }{2} t} (1 + p[ e^{i \hbar \omega t} - 1] )^n.
\end{align}
Making the substitution $p = \lambda / n$ we find
\begin{align}
    \varphi_{n,p} (t) &= e^{ i\frac{\hbar \omega }{2} t} (1 + \frac{\lambda}{n}[ e^{i \hbar \omega t} - 1] )^n
\end{align}
However in the limit we have that $\lim_{n \to \infty} (1 + \frac{x}{n})^n = e^{x}$ and therefore
\begin{align}
    \lim_{\substack{n \to \infty \\ np = \lambda}}\varphi_{n,p} (t) &= e^{\lambda ( e^{i \hbar \omega t} - 1) + i\frac{\hbar \omega }{2} t} .
\end{align}
If these are equal for all values of $t$ we deduce that up to arbitrary phases,
\begin{equation}
    \lim_{\substack{n \to \infty \\ np = \lambda}} |n,p\> = |\sqrt{\lambda} \>
\end{equation}
where $|\sqrt{\lambda}\>$ is a coherent state.

These results are enough to prove convergence of the binomial state fluctuation relation to the coherent state fluctuation relation. The quantum distortion factors can also be obtained by perturbative means or by using the relevant quantities in the coherent state limit.







\end{document}